\documentclass[sn-mathphys]{sn-jnl}

\jyear{2023}%

\usepackage{xcolor,mathtools,amsthm,amsmath,amssymb,amsbsy,amsfonts,latexsym,enumitem,amsopn,amstext,amsxtra,euscript,amscd,commath,mathtools,braket,lineno,hyperref,cite,amsbsy}

\def\II{{\mathcal I}}

\def\Z{{\mathbb Z}}

\def\00{{\bf 0}}
\def\11{{\bf 1}}
\def\+{\oplus}

\theoremstyle{thmstyleone}%
\newtheorem{theorem}{Theorem}
\newtheorem{corollary}[theorem]{Corollary}

\theoremstyle{thmstyletwo}%
\newtheorem{remark}{Remark}%

\theoremstyle{thmstylethree}%
\newtheorem{definition}{Definition}%

\begin{document}

\title{Lattice attack on group ring NTRU: The case of the dihedral group}


\author[1]{\fnm{Vikas} \sur{Kumar}}\email{v\_kumar@ma.iitr.ac.in}

\author[2]{\fnm{Ali} \sur{Raya}}\email{ali\_r@cs.iitr.ac.in}

\author[2]{\fnm{Sugata} \sur{Gangopadhyay}}\email{sugata.gangopadhyay@cs.iitr.ac.in}

\author*[1]{\fnm{Aditi Kar} \sur{Gangopadhyay}}\email{aditi.gangopadhyay@ma.iitr.ac.in}

\affil*[1]{\orgdiv{Department of Mathematics}, \orgname{IITR}, \orgaddress{\street{Roorkee}, \city{Haridwar}, \postcode{247667}, \state{Uttarakhand}, \country{India}}}

\affil[2]{\orgdiv{Department of Computer Science and Engineering}, \orgname{IITR}, \orgaddress{\street{Roorkee}, \city{Haridwar}, \postcode{247667}, \state{Uttarakhand}, \country{India}}}


\abstract{Group ring NTRU (GR-NTRU) provides a general structure to design different variants of NTRU-like schemes by employing different groups. Although, most of the schemes in literature are built over cyclic groups, nonabelian groups can also be used. Coppersmith and Shamir in 1997 have suggested that noncommutativity may result in better security against some lattice attacks for some groups.
Lattice attacks on the public key of NTRU-like cryptosystems try to retrieve the private key by solving the shortest vector problem (SVP) or its approximation in a lattice of a certain dimension, assuming the knowledge of the public key only.
 This paper shows that dihedral groups do not  guarantee better security against this class of attacks. We prove that retrieving the private key is possible by solving the SVP in two lattices with half the dimension of the original lattice generated for GR-NTRU based on dihedral groups. The possibility of such an attack was mentioned by Yasuda et al. 
 (IACR/2015/1170). In contrast to their proposed approach, we explicitly provide the lattice reduction without any structure theorem from the 
 representation theory for finite groups. Furthermore, we 
 demonstrate the effectiveness of our technique with experimental results.}




\maketitle

\section{Introduction}\label{sec1}

The first NTRU cryptosystem \citep{ntru-first-paper} was proposed early in 1996 as a public key scheme built over a quotient ring of polynomials. Being an efficient scheme with reasonable memory requirements, NTRU has attracted cryptanalysts and undergone extensive analysis for its security and performance. 
IEEE considered NTRU for standardization as an efficient scheme based on post-quantum mathematical problems (IEEE-1363.1) \citep{IEEE}.
Moreover, a few NTRU-like schemes have been submitted to the National Institute of Standards and Technology (NIST) competition and proceeded through the different rounds of evaluation \citep{NTRUEncrypt,pqNTRU,HRSS-round1,NTRUprime-round1,NTRUprime-round2,NTRU-round2,NTRU-round3,nist-third-report}. 


For most submissions of NTRU, either in the literature or NIST's competition, the underlying ring $\mathcal{R}$ is selected to be a commutative ring, for example, $\mathcal{R} = \mathbb{Z}_q[x]/(x^N-1) \text{ for prime $N$ or } \mathcal{R} = \mathbb{Z}_q[x]/(x^{2^n}+1)$  for a positive integer $n$. However, when Coppersmith and  Shamir established their lattice attack against NTRU \citep{Coppersmith-Shamir-paper}, they stated that considering a noncommutative group algebra could be another direction to provide better security against their attack.

Few variants of NTRU based on noncommutative rings have been introduced and studied. 
In 1997, Hoffstein and Silverman \citep{non-commutative-ntru-unpublished} proposed a noncommutative variant of NTRU based on the dihedral group, which was broken soon by Coppersmith \citep{ibm-report}. The same design of the noncommutative NTRU based on the dihedral group has been analyzed by Truman \citep{Marylandthesis}. In the same work Truman extends the idea of the noncommutative NTRU to other group rings showing that Coppersmith's attacks can only work for the group ring based on the dihedral or closely related group rings. 

Another attempt to build a noncommutative scheme analogous to NTRU but based on  Quaternion algebra was proposed by Malekian et al. \citep{QTRU}. According to the authors' claim, the proposed cryptosystem is multidimensional and more resistant to some lattice attacks due to the noncommutativity of the underlying algebraic structure.
In \citep{Sakurai}, Yasuda et al. describe group ring NTRU (GR-NTRU), which serves as a general structure to build different variants of NTRU-like schemes. The group ring $\Z G$ corresponding to a finite group $G$ is used to create an NTRU-like variant, where the group $G$ can be abelian or nonabelian. They generalize the attack by Gentry \citep{Gentry} against composite-degree NTRU to GR-NTRU using the concepts of group representation theory.
Furthermore, they discuss their attack against some groups in the context of GR-NTRU, including the dihedral group.
It is worth mentioning here that the schemes discussed in ~\citep{non-commutative-ntru-unpublished, Marylandthesis} are variants of NTRU using group ring based on the dihedral group. However, the designs of these schemes differ from the dihedral group based GR-NTRU, and the attack proposed by Coppersmith in \citep{ibm-report} can not be applied against GR-NTRU. 

\noindent \textbf{Our contribution:}
For GR-NTRU on a dihedral group of order $2N$, Yasuda et al.~\citep{Sakurai} provided an overview of a lattice reduction, and an estimate of the lattice reduction complexity. However, they do not provide any explicit mapping 
or concrete algorithm for their reduction. Our work in this paper: 
\begin{itemize}
    \item  explicitly shows the lattice reduction using simple matrix algebra; 
    \item explains how to map the problem of retrieving the private key from solving the SVP (or an approximation of it) in a $4N$-dimensional lattice into two smaller $2N$-dimensional lattices. Furthermore, the structures of the smaller lattices are provided; 
    \item provides a pull-back approach to retrieve two decryption keys; one is a short-enough (non ternary) key while the other is a ternary key; 
    \item supports the reduction's correctness through theoretical analysis and experimental results; 
    \item proves that the dihedral group does not provide additional security to GR-NTRU compared to the standard NTRU based on a cyclic group of order $N$. 
\end{itemize}

The remaining paper is structured as follows: Section \ref{group rings} provides preliminaries related to group rings and the matrix representation of group ring elements. Section \ref{NTRU} describes the general design of GR-NTRU and the lattice attack on the public key. Section \ref{sec-4} discusses GR-NTRU based on the dihedral group and provides our method for lattice reduction. Section \ref{results} presents experimental verification of our lattice reduction attack on this scheme. Finally, we conclude our work in Section \ref{conclusion}.

\section{Group rings}
\label{group rings}
Group rings can be defined for arbitrary groups. However, as far as the scope of this work is concerned, we will consider group rings of finite groups only.
For a ring $R$ and a finite group $G = \{g_i : i = 1,2,\dots,n\}$ of order $n$, with identity element $g_1 = e$, define the set of formal sums
    \begin{equation}
        RG = \Biggl\{{a} = \sum_{i=1}^{n}\alpha_ig_i : \alpha_i\in R\text{~~~for~~~} i = 1,2,\dots,n\Biggr\}.
    \end{equation} 
Suppose ${a} = \sum_{i=1 }^{n}\alpha_ig_i $ and ${b}= \sum_{i=1}^{n}\beta_ig_i $ in $RG$. By definition, ${a} = {b}$ if and only if $\alpha_i = \beta_i$ for all $i = 1,2,\dots,n$.
Sum of ${a}$ and ${b}$ is defined as:
\begin{equation}
\label{sum}
    {a} + {b} = \sum_{i=1}^{n}\alpha_ig_i + \sum_{i=1}^{n}\beta_ig_i = \sum_{i=1}^{n}(\alpha_i+\beta_i)g_i 
\end{equation}
Define the product of ${a}$ and ${b}$ as:
\begin{equation}
\label{prod}
   ab = \left(\sum_{i=1}^{n}\alpha_ig_i \right)\left(\sum_{i=1}^{n}\beta_ig_i\right) = \sum_{i=1}^{n}\gamma_ig_i 
\end{equation}
where
\begin{equation}\label{prod-coeff}
    \gamma_i = \sum_{g_hg_k = g_i}\alpha_h\beta_k.
\end{equation}
For each element $a = \sum_{i=1}^{n}\alpha_ig_i \in RG $, we associate a unique vector 
${\textbf{a}} = (\alpha_1,\alpha_2,\dots,\alpha_{n})$. We use $a$ and $\mathbf{a}$ 
interchangeably to refer to an element of group ring $RG$. In vector notation
$${\textbf{a}}+{\textbf{b}} = (\alpha_1+\beta_1,\alpha_2+\beta_2,\dots,\alpha_n+\beta_n),~~ {\textbf{a}}\star {\textbf{b}} = (\gamma_1,\gamma_2,\dots,\gamma_n)$$
where $\gamma_i$, for $i=1, 2, \ldots, n$, are given by \eqref{prod-coeff}, denote
coordinatewise addition and the convolutional product of two vectors ${\textbf{a}},{\textbf{b}}\in RG$, respectively.

\begin{definition}(\citep[Chapter 3]{MS})
    The set $RG$ together with the operations defined in \eqref{sum} and \eqref{prod} forms a ring. We say that $RG$ is the \textit{group ring} of $G$ over $R$.
\end{definition}

 Suppose $R$ is ring with unity $1_{R}$, then $\textbf{1}_{RG} = (1_R,0,0,\dots,0)$ is the unity of the group ring $RG$.
We define the scalar product of elements of $RG$ with elements $\delta\in R$ as follows:
\begin{equation}
    \delta {\textbf{a}} = \delta(\alpha_1,\alpha_2,\dots,\alpha_{n}) = (\delta\alpha_1,\delta\alpha_2,\dots,\delta\alpha_{n}).
\end{equation}
This makes $RG$ an $R$-module. Further, if $R$ is commutative then $RG$ is an $R$-algebra.
\begin{definition}
    Let the vector $\normalfont{\textbf{a}}=(\alpha_1,\alpha_2,\alpha_3,\ldots,\alpha_n)\in RG$, and $1\leq r\leq n-1$. Then, $\normalfont{\textbf{a}}^{(r)} = (\alpha_r,\alpha_{r+1},\ldots,\alpha_{r-1})$ denotes the rotation of $\normalfont{\textbf{a}}$ to left by $r$ positions and $\normalfont{\textbf{a}}^{(-r)} = (\alpha_{n-r+1},\alpha_{n-r+2},\ldots,\alpha_{n-r})$ denote the rotation of $\normalfont{\textbf{a}}$ to right by $r$ positions, and $\normalfont{{\textbf{a}}}^{(0)}={\textbf{a}}$. We may also write $\normalfont{\textbf{a}}^{(-r)}=\normalfont{\textbf{a}}^{(n-r)}$.
\end{definition}
\textbf{Matrix representation of group ring elements: }
In \citep{Hurley}, Hurley establishes an isomorphism between a group ring $RG$ and a certain subring of $n \times n$ matrices over $R$. For a group $G = \{g_1,g_2,\ldots,g_n\}$, define the matrix of $G$ as 
\begin{equation}\label{matrix_rep}
    M_{G} = \begin{pmatrix}
g_1^{-1}g_1 & g_1^{-1}g_2 &\ldots\ldots& g_1^{-1}g_n\\
g_2^{-1}g_1 & g_2^{-1}g_2 &\ldots\ldots& g_2^{-1}g_n\\
\vdots & \vdots & \ddots & \vdots\\
g_n^{-1}g_1 & g_n^{-1}g_2 &\ldots\ldots& g_n^{-1}g_n\\
\end{pmatrix}.
\end{equation}
We now construct the $RG$-matrix of an element ${\textbf{a}}=(\alpha_{g_1},\alpha_{g_2},\dots,\alpha_{g_n})\in RG$ as follows
\begin{equation}
    M_{RG}(\textbf{a}) = \begin{pmatrix}
\alpha_{g_1^{-1}g_1} & \alpha_{g_1^{-1}g_2} & \ldots\ldots & \alpha_{g_1^{-1}g_n}\\
\alpha_{g_2^{-1}g_1} & \alpha_{g_2^{-1}g_2} & \ldots\ldots & \alpha_{g_2^{-1}g_n}\\
\vdots & \vdots & \ddots & \vdots\\
\alpha_{g_n^{-1}g_1} & \alpha_{g_n^{-1}g_2} &\ldots\ldots& \alpha_{g_n^{-1}g_n}\\
\end{pmatrix}.
\end{equation}

The set $M_{RG} = \{M_{RG}(\textbf{a}) : \textbf{a} \in RG\}$ is the subring of the ring of $n\times n$ matrices over $R$, denoted by $M_n(R)$. We say a matrix $A\in M_n(R)$ is an $RG$-matrix if there is an $\textbf{a}\in RG$ such that $A = M_{RG}(\textbf{a})$. 
\begin{theorem}(\citep[Thereom 1]{Hurley})
    The mapping $\tau : RG \rightarrow M_{RG}\subset M_n(R)$ defined as $\tau(\normalfont \textbf{a}) = M_{RG}(\textbf{a})$ is a bijective ring homomorphism, i.e., $\tau(\normalfont\textbf{a}+\textbf{b}) = \tau(\textbf{a})+\tau(\textbf{b}) = M_{RG}(\textbf{a}) + M_{RG}(\textbf{b})$, and $\tau(\normalfont\textbf{a}\star\textbf{b}) = \tau(\textbf{a})\cdot\tau(\textbf{b}) = M_{RG}(\textbf{a})\cdot M_{RG}(\textbf{b})$, where $+,\cdot$ denote the usual matrix addition and multiplication, respectively. Furthermore, $\tau$ is a module $R$-homomorphism, i.e., $\tau(\delta \normalfont\textbf{a}) = \delta\tau(\textbf{a}) = \delta M_{RG}(\textbf{a})$, for $\delta\in R$.
\end{theorem}
\begin{theorem}(\citep[Thereom 2]{Hurley})
    Let $R$ be a ring with unity and $G$ be a finite group. Then, $\normalfont\textbf{a}\in RG$ is a unit if and only if $M_{RG}(\normalfont\textbf{a})$ is invertible in $M_n(R)$. In that case, inverse of $M_{RG}(\normalfont\textbf{a})$ is also an $RG$-matrix.
\end{theorem}

\begin{corollary}(\citep[corollary 2]{Hurley})
When $R$ is commutative ring with unity, element $\normalfont\textbf{a}$ is a unit in $RG$ if and only if $det(\tau(\normalfont\textbf{a}))$ is a unit in $R$. In case, when $R$ is a field then $\normalfont\textbf{a}$ is a unit if and only if $det(\tau(\normalfont\textbf{a}))\neq 0$.
\end{corollary}

\section{Group ring NTRU (GR-NTRU)}
\label{NTRU}
Yasuda et al. \citep{Sakurai} proposed a general framework to develop NTRU-like cryptosystems based on group rings. They call this scheme: Group ring NTRU or GR-NTRU. The idea for such a construction is as follows:\\

\textbf{Parameters selection: }
Let $n, p, q, d$ be positive integers with $p$ prime, $p << q$, gcd$(p,q) = 1$, $2d+1\leq n$, and $q>(6d+1)p$. Throughout this paper, we will take $p=3, d$ is taken to be at most $\lfloor \frac{n}{3}\rfloor$, and $q$ is usually a power of 2. Let $\Z G, \Z_qG$, and $\Z_pG$ be group rings where $\Z, \Z_q$, and $\Z_p$ denote the ring of integers, ring of integers modulo $q$, and ring of integers modulo $p$, respectively, and $G$ is a finite group of order $n$.

Let $t_1,t_2$ be positive integers such that $t_1+t_2\leq n$. Define 
$$\mathcal{P}(t_1,t_2) =
   \left\{ {\textbf{a}} \in \Z G  \middle\vert \begin{array}{l}
    {\textbf{a}}~~\text{has $t_1$ coefficients equal to $1$} \\
    {\textbf{a}}~~\text{has $t_2$ coefficients equal to $-1$} \\
    \text{other coefficients are 0}
  \end{array}\right\}.$$
Elements in $\mathcal{P}(t_1,t_2)$ are referred to as ternary vectors or ternary elements. For ${\textbf{a}}\in \Z_q G$, the \textit{centered lift} of ${\textbf{a}}$ is the unique element ${\textbf{a}}_{lifted}\in \Z G$ whose coefficients are in the interval $\left(-\frac{q}{2},\frac{q}{2}\right]$ and ${\textbf{a}}_{lifted}\pmod q = {\textbf{a}}$, where ${\textbf{a}}_{lifted}\pmod q$ is obtained by reducing each coefficient of the vector ${\textbf{a}}_{lifted}$ modulo $q$. \\
A message is a vector in $\Z G$ that is the centered lift of some element in $\Z_p G$.  In other words,  message space consists of elements from $\Z G$ whose coefficients are between $-\frac{p}{2}$ and $\frac{p}{2}$.\\

\textbf{Key generation: }
\begin{enumerate}[label=(\roman*)]
    \item choose ${\textbf{f}}\in \mathcal{P}(d+1,d)$ such that there exist ${\textbf{f}}_q\in \Z_q G$, ${\textbf{f}}_p \in \Z_p G$ satisfying ${\textbf{f}}\star {\textbf{f}}_q \equiv {1}_{\Z_q G}\pmod q$ and ${\textbf{f}} \star {\textbf{f}}_p \equiv {1}_{\Z_p G} \pmod p$. 
    \item choose another element ${\textbf{g}}\in \mathcal{P}(d,d)$.
    \item construct ${\textbf{h}}\in\Z_q G$ such that ${\textbf{f}} \star {\textbf{h}} = {\textbf{g}} \pmod q$, equivalently ${\textbf{h}}={\textbf{f}}_q\star{\textbf{g}}\pmod q$.
    \item declare ${\textbf{h}},p,q$ to be public key.
    \item $({\textbf{f}},{\textbf{g}})$ and ${\textbf{f}}_p$ are kept private.\\
\end{enumerate}

\textbf{Encryption: }To encrypt a message ${\textbf{m}}$, we first randomly choose ${\textbf{r}}\in\mathcal{P}(d,d)$. Then, the ciphertext is computed as follows:
$
    {\textbf{c}} = p{\textbf{h}}\star {\textbf{r}} + {\textbf{m}} \pmod q.
$\\

\textbf{Decryption: }First, compute ${\textbf{a}} \equiv {\textbf{f}}\star {\textbf{c}} \pmod q$. Then, centerlift it to ${\textbf{a}}_{lifted}$ modulo $q$. Now, ${\textbf{m}}$ can be recovered by computing ${\textbf{f}}_p\star {\textbf{a}}_{lifted} \pmod p$ and centerlifting it modulo $p$.\\

\textbf{Correctness: }We have ${\textbf{a}} \equiv p{\textbf{g}}\star{\textbf{r}} + {\textbf{f}}\star {\textbf{m}}\pmod q$. Since ${\textbf{f}}\in\mathcal{P}(d+1,d)$,
${\textbf{g}}, {\textbf{r}}\in\mathcal{P}(d,d)$, and coefficients of ${\textbf{m}}$ lie between $-\frac{p}{2}$ to $\frac{p}{2}$. Therefore, the largest coefficient of ${\textbf{g}}\star {\textbf{r}}$ can be $2d$ and the largest coefficient of ${\textbf{f}}\star {\textbf{m}}$ can be $(2d+1)\frac{p}{2}$. Consequently, the largest coefficient of $p{\textbf{g}}\star {\textbf{r}} + {\textbf{f}}\star {\textbf{m}}$ is at most $(6d+1)\frac{p}{2}$. Thus, if $q>(6d+1)p$, computing ${\textbf{a}}\equiv {\textbf{f}}\star {\textbf{c}}\pmod q$ and then centerlifting it gives exactly the element $p{\textbf{g}}\star {\textbf{r}} + {\textbf{f}}\star {\textbf{m}}$ without modulo $q$. Now, we multiply this element with ${\textbf{f}}_p$ and reduce coefficients modulo $p$ to recover an element in $\Z_p G$ whose centered lift gives the message ${\textbf{m}}$.\\

\textbf{Lattice attack on GR-NTRU cryptosystem: }Let 
 ${\textbf{h}} =(h_1,h_2,\dots,h_{n})$ be a GR-NTRU public key generated by the private key ${\textbf{f}} =(f_1,f_2,\dots,f_{n})$ and ${\textbf{g}} =(g_1,g_2,\dots,g_{n})$, i.e., ${\textbf{f}}\star{\textbf{h}}={\textbf{g}}\pmod q$.
 
The NTRU lattice $L_{{\textbf{h}}}$ associated to ${\textbf{h}}$ is a $2n$-dimensional lattice generated by the rows of the matrix 
\begin{equation}
    M_{{\textbf{h}}} = \begin{pmatrix}
        I_n & H\\
        \00_n & qI_n
    \end{pmatrix}
\end{equation}
where, $H=\tau({\textbf{h}})$ is a matrix of order $n$ with vector ${\textbf{h}}$ as its first row, and $I_n$ is identity matrix of order $n$.
 \begin{theorem}
 \label{SVP}
The vector $({\normalfont\textbf{f}},{\normalfont\textbf{g}})$ lies in the lattice $L_{{\normalfont\textbf{h}}}$, and if $n$ is large, then there is a high probability that the shortest nonzero vectors in $L_{{\normalfont\textbf{h}}}$ are $(\normalfont{\textbf{f}},{\textbf{g}})$ and other vectors obtained by ``rotations" of $\normalfont{{\textbf{f}}}$ and $\normalfont{{\textbf{g}}}$.
\end{theorem}
Therefore, the private key can be recovered by solving the SVP in a lattice of dimension $2n$. The proof of the Theorem \ref{SVP} can be given precisely the same way as in \citep[Proposition 6.59, 6.61]{HPS}.
By ``rotation" in the above theorem, we refer to a transformation related to the underlying group and not 
necessarily a cyclic rotation in the usual sense. We discuss the particular transformation and provide
the proof for the dihedral group based GR-NTRU in Section~\ref{short_vector}.
\begin{remark}\label{remark1}
    If one can find a pair of vectors $(\normalfont{\textbf{f}}^{\prime},{\textbf{g}}^{\prime})\in L_{{\textbf{h}}}$, not necessarily the private key, with small enough coefficients such that for an arbitrary message $\normalfont{\textbf{m}}$ encrypted using the public key $\normalfont{\textbf{h}}$ and a random ternary vector $\normalfont{\textbf{r}}$, the largest coefficient of $\normalfont p{\textbf{g}}^{\prime}\star{\textbf{r}}+{\textbf{f}}^{\prime}\star{\textbf{m}}$ is at most $\frac{q}{2}$. Then, $(\normalfont{\textbf{f}}^{\prime},{\textbf{g}}^{\prime})$ serves the purpose of decryption for $\normalfont{\textbf{m}}$. However, this requires $\normalfont{\textbf{f}}^{\prime}$ to be invertible over $\Z_p G$. 
\end{remark}

\textbf{NTRU as a special case of GR-NTRU: }
It is straightforward to observe that the NTRU scheme \citep[Chapter 6.10]{HPS} can be reformulated over the group ring:
\begin{equation}
  \Z_q C_N \cong \Z_q[x]/(x^N-1)
\end{equation}
where $C_N$ is a cyclic group of order $N$. In this case, $H$ is a circulant matrix of order $N$. Further, the lattice attack for GR-NTRU is the generalization of lattice attack on NTRU given in \citep{ntru-first-paper}. Also, Gentry in \citep{Gentry} proposed an attack on NTRU with composite value of $N$. Therefore, the  value of $N$ is always taken to be prime.

\section{GR-NTRU based on the dihedral group}\label{sec-4}
The dihedral group $D_N$ of order $2N$ is given by $D_N = \bigl<x , y : x^N = y^2 = 1, xy = yx^{-1}\bigr>$, i.e., $D_N = \{1,x,\dots,x^{N-1},y,yx,\dots,yx^{N-1}\}$.

This article will focus on the GR-NTRU built over group ring $\Z_q D_N$:
 
 \begin{equation}
     \Z_q D_N \cong \frac{\Z_q[x,y]}{\big<x^N-1,y^2-1,yx - x^{N-1}y\big>} 
 \end{equation}
According to Theorem \ref{SVP}, finding the private key in GR-NTRU cryptosystem over $\Z_q D_N$ is equivalent to solve the SVP in a $4N$-dimensional lattice. This section highlights our reduction from solving SVP in a $4N$-dimensional lattice to two $2N$-dimensional lattices.
\begin{theorem}\label{matrixofdihedral}
    Consider an element $h\in\Z D_N$ where $$h = h_{00}1+h_{01}x+\dots+h_{0N-1}x^{N-1}+h_{10}y+h_{11}yx+\dots+h_{1N-1}yx^{N-1}.$$  Let $\normalfont{{\textbf{h}}}_0 =  (h_{00},h_{01},\ldots,h_{0N-1})$ and $\normalfont{\textbf{h}}_1 =  (h_{10},h_{11},\ldots,h_{1N-1})$, so that
$\normalfont{\textbf{h}} = ({\textbf{h}_0},{\textbf{h}_1})$
is a $2N$-length vector. Then,
\begin{enumerate}
    \item The matrix of the group $D_N$, and consequently the $\Z D_N$-matrix of $\normalfont{\textbf{h}}$ is a $2N\times 2N$ matrix of the form $H=\tau(\normalfont{\textbf{h}})= \begin{pmatrix}
    H_0 & H_1 \\
    H_1 & H_0
\end{pmatrix}$, where $H_0$ is a circulant matrix with vector $\normalfont{\textbf{h}_0}$, and $H_1$ is a reverse circulant matrix with vector $\normalfont{\textbf{h}_1}$, as their first rows, respectively.

\item Let $L_{\normalfont{\textbf{h}}}$ denotes the lattice associated to $\normalfont{\textbf{h}}$ spanned by the rows of the matrix 

\[
       M_{\normalfont{\textbf{h}}} = \left[\begin{array}{c|c}
        I_{2N} & H \\
        \hline
         \00_{2N} &  qI_{2N} \\
    \end{array}\right]
\] 

   and $\normalfont{{\textbf{f}}}=({\textbf{f}}_0,{\textbf{f}}_1),\normalfont{{\textbf{g}}}=({\textbf{g}}_0,{\textbf{g}}_1)$ are vectors such that $\normalfont{({\textbf{f}},{\textbf{g}})}\in L_{{\normalfont\textbf{h}}}$. Then, for every $-N+1\leq r\leq N-1$, $(\normalfont{{\textbf{f}}_0^{(r)},{\textbf{f}}_1^{(-r)},{\textbf{g}}_0^{(r)},{\textbf{g}}_1^{(-r)}}), (\normalfont{{\textbf{f}}_1^{(r)},{\textbf{f}}_0^{(-r)},{\textbf{g}}_1^{(r)},{\textbf{g}}_0^{(-r)}})\in L_{\normalfont{{\textbf{h}}}}$
\end{enumerate} 
\end{theorem}
\begin{proof}
The proof of the first part can be derived directly from \eqref{matrix_rep}.
Now suppose that $(\normalfont{{\textbf{f}}_0,{\textbf{f}}_1,{\textbf{g}}_0,{\textbf{g}}_1})\in  L_{\normalfont{{\textbf{h}}}}$ then there exists a vector $(\normalfont{{\textbf{u}}_0,{\textbf{u}}_1})$ with integer entries such that $(\normalfont{{\textbf{f}}_0,{\textbf{f}}_1,{\textbf{u}}_0,{\textbf{u}}_1})M_{{\textbf{h}}}=(\normalfont{{\textbf{f}}_0,{\textbf{f}}_1,{\textbf{g}}_0,{\textbf{g}}_1}).$ It is easy to check that $(\normalfont{{\textbf{f}}_1,{\textbf{f}}_0,{\textbf{u}}_1,{\textbf{u}}_0})M_{{\textbf{h}}}=(\normalfont{{\textbf{f}}_1,{\textbf{f}}_0,{\textbf{g}}_1,{\textbf{g}}_0})$. Then the second part follows from the observations that, since $H_0$ is a circulant matix and $H_1$ is a reverse circulant matrix, therefore for every $0\leq r\leq N-1$ and a vector $\normalfont{\textbf{a}}$ of length $N$ with integer entries, $\normalfont{\textbf{a}}^{(r)}H_0 = ({\textbf{a}}H_0)^{(r)}$ and $\normalfont{\textbf{a}}^{(r)}H_1 = ({\textbf{a}}H_1)^{(-r)}$. Also, $\normalfont{\textbf{a}}^{(-r)}H_0 = ({\textbf{a}}H_0)^{(-r)}$ and $\normalfont{\textbf{a}}^{(-r)}H_1 = ({\textbf{a}}H_1)^{(r)}$.
\end{proof}
\subsection{Estimation of lengths of short vectors in the lattice $L_{{\textbf{h}}}$}\label{short_vector}
Let ${\textbf{f}} = ({\textbf{f}}_0,{\textbf{f}}_1)\in\mathcal{P}(d+1,d)$, ${\textbf{g}} = ({\textbf{g}}_0,{\textbf{g}}_1)\in \mathcal{P}(d,d)$ be two randomly and uniformly generated ternary vectors with $d$ at most $\lfloor\frac{2N}{3}\rfloor$. 
Therefore, for $-N+1\leq r\leq N-1$, the length of the vectors $(\normalfont{{\textbf{f}}_0^{(r)},{\textbf{f}}_1^{(-r)},{\textbf{g}}_0^{(r)},{\textbf{g}}_1^{(-r)}})$ and $ (\normalfont{{\textbf{f}}_1^{(r)},{\textbf{f}}_0^{(-r)},{\textbf{g}}_1^{(r)},{\textbf{g}}_0^{(-r)}})$ is  at most $\sqrt{\frac{8N}{3}+1}\approx 1.63\sqrt{N}.$ 
Let ${\textbf{h}}$ be a public key constructed from private vectors ${\textbf{f}}$ and ${\textbf{g}}$, i.e., ${\textbf{h}}={\textbf{f}}_q\star{\textbf{g}}\pmod q$.  According to Gaussian heuristic estimation, the length of the shortest vector in the lattice $L_{{\textbf{h}}}$ is $$\sigma(L_{{\textbf{h}}}) = \sqrt{\frac{4N}{2\pi e}}(\det L_{{\textbf{h}}})^{\frac{1}{4N}} = \sqrt{\frac{2qN}{\pi e}}\approx \sqrt{\frac{8}{\pi e}}N\approx 0.97N,$$ since $q\approx 4N$. Furthermore, 
$$\frac{\norm{(\normalfont{{\textbf{f}}_0^{(r)},{\textbf{f}}_1^{(-r)},{\textbf{g}}_0^{(r)},{\textbf{g}}_1^{(-r)}})}}{\sigma(L_{{\textbf{h}}})} =  \frac{\norm{(\normalfont{{\textbf{f}}_1^{(r)},{\textbf{f}}_0^{(-r)},{\textbf{g}}_1^{(r)},{\textbf{g}}_0^{(-r)}})}}{\sigma(L_{{\textbf{h}}})} \approx \frac{1.68}{\sqrt{N}}.$$ So, these vectors are a factor of $O\left(\frac{1}{\sqrt{N}}\right)$ shorter than predicted by the
Gaussian heuristic. Hence, for larger values of $N$, there is a high probability that the vectors $(\normalfont{{\textbf{f}}_0^{(r)},{\textbf{f}}_1^{(-r)},{\textbf{g}}_0^{(r)},{\textbf{g}}_1^{(-r)}}), (\normalfont{{\textbf{f}}_1^{(r)},{\textbf{f}}_0^{(-r)},{\textbf{g}}_1^{(r)},{\textbf{g}}_0^{(-r)}})$ are the shortest vectors in the lattice $L_{{\textbf{h}}}$.
\subsection{Our lattice reduction for GR-NTRU over $\Z_q D_N$}\label{our-reduction}
Let $
\mathcal{I} = \begin{pmatrix}
    I_N & ~~~I_N \\
I_N & -I_N
\end{pmatrix}$ where $I_N$ is an $N\times N$ identity matrix. For a ring $R$ (in our case $R=\Z$) with characteristic not $2$, $\mathcal{I}$ is invertible over $R$ or over the field of quotients of $R$.  Conjugating $\tau({\textbf{h}})$ by $\mathcal{I}$, we get 
\begin{equation}\label{conjugate}
    \mathcal{I}\begin{pmatrix}
    H_0 & H_1 \\
H_1 & H_0
\end{pmatrix}\mathcal{I}^{-1} = \begin{pmatrix}
    H_0 + H_1 & \00_N \\
\00_N & H_0 - H_1
\end{pmatrix}.
\end{equation}
Here, the matrix $\mathcal{I}$ is independent of $H_0$ and $H_1$.
\begin{theorem}
    Let the vector $(\normalfont{\textbf{f}},{\textbf{g}})$ where $\normalfont{\textbf{f}} = ({\textbf{f}}_0,{\textbf{f}}_1)$ and $\normalfont{\textbf{g}} = ({\textbf{g}}_0,{\textbf{g}}_1)$ lies in the lattice $L_{\normalfont{\textbf{h}}}$, i.e., $ \normalfont{\textbf{f}}\star {\textbf{h}} + q{\textbf{u}} = {\textbf{g}}$ for some vector $\normalfont{\textbf{u}}$. Let $-N+1\leq r\leq N-1$, then the vector $(\normalfont{\textbf{f}}_0^{(r)}+{\textbf{f}}_1^{(-r)},{\textbf{g}}_0^{(r)}+{\textbf{g}}_1^{(-r)})\in L_{{\textbf{h}}_0+{\textbf{h}}_1}$, and the vectors $(\normalfont{\textbf{f}}_0^{(r)}-{\textbf{f}}_1^{(-r)},{\textbf{g}}_0^{(r)}-{\textbf{g}}_1^{(-r)}),(\normalfont{\textbf{f}}_1^{(r)}-{\textbf{f}}_0^{(-r)},{\textbf{g}}_1^{(r)}-{\textbf{g}}_0^{(-r)})\in L_{{\textbf{h}}_0-{\textbf{h}}_1}$.
\end{theorem}
\begin{proof}
We have $\normalfont{\textbf{f}}\star {\textbf{h}} + q{\textbf{u}} = {\textbf{g}}$. Applying $\tau$ to both sides and using the fact that $\tau$ is a ring homomorphism as well as a $\Z$-module homomorphism, and using Equation \eqref{conjugate}, we get 

\begin{align*}
    \tau({\textbf{f}})\cdot \tau({\textbf{h}}) + q \tau({\textbf{u}})&= \tau({\textbf{g}})\\
     \begin{pmatrix}
         F_0 & F_1\\
         F_1 & F_0
     \end{pmatrix}   \begin{pmatrix}
         H_0 & H_1\\
         H_1 & H_0
     \end{pmatrix} + q \begin{pmatrix}
         U_0 & U_1\\
         U_1 & U_0
     \end{pmatrix} &= \begin{pmatrix}
         G_0 & G_1\\
         G_1 & G_0
     \end{pmatrix}\\
\II\begin{pmatrix}
         F_0 & F_1\\
         F_1 & F_0
     \end{pmatrix}\II^{-1}\II \begin{pmatrix}
         H_0 & H_1\\
         H_1 & H_0
     \end{pmatrix}\II^{-1} + q\II  \begin{pmatrix}
         U_0 & U_1\\
         U_1 & U_0
     \end{pmatrix}\II^{-1} &=  \II \begin{pmatrix}
         G_0 & G_1\\
         G_1 & G_0
     \end{pmatrix} \II^{-1}\\
\begin{pmatrix}
         F_0 + F_1 & \00_N\\
         \00_N & F_0 - F_1
     \end{pmatrix}\begin{pmatrix}
         H_0 + H_1 & \00_N\\
         \00_N & H_0 - H_1
     \end{pmatrix} + q&\begin{pmatrix}
         U_0 + U_1 & \00_N\\
         \00_N & U_0 - U_1
     \end{pmatrix} \\= &\begin{pmatrix}
         G_0 + G_1 & \00_N\\
         \00_N & G_0 - G_1
     \end{pmatrix}.
\end{align*}
Equivalently,
\begin{align*}
    (F_0+ F_1)(H_0+ H_1) +q(U_0+ U_1) &= (G_0+ G_1)\\
    (F_0- F_1)(H_0- H_1) +q(U_0- U_1) &= (G_0- G_1).
\end{align*}
Considering the first rows, we have
\begin{align*}
({\textbf{f}}_0+{\textbf{f}}_1)(H_0+ H_1) +q({\textbf{u}}_0+{\textbf{u}}_1) &= ({\textbf{g}}_0+{\textbf{g}}_1)\\
({\textbf{f}}_0-{\textbf{f}}_1)(H_0- H_1) +q({\textbf{u}}_0-{\textbf{u}}_1) &= ({\textbf{g}}_0-{\textbf{g}}_1).
\end{align*}
Therefore,
\begin{align*}
({\textbf{f}}_0+{\textbf{f}}_1,{\textbf{u}}_0+{\textbf{u}}_1)\begin{pmatrix}
    I_{2N} & H_0 + H_1\\
    \00_N & qI_{2N}
\end{pmatrix} &= ({\textbf{f}}_0+{\textbf{f}}_1,{\textbf{g}}_0+{\textbf{g}}_1)\\
({\textbf{f}}_0-{\textbf{f}}_1,{\textbf{u}}_0-{\textbf{u}}_1)\begin{pmatrix}
    I_{2N} & H_0 - H_1\\
    \00_N & qI_{2N}
\end{pmatrix} &= ({\textbf{f}}_0-{\textbf{f}}_1,{\textbf{g}}_0-{\textbf{g}}_1).
\end{align*}
\noindent This implies that the vector $({\textbf{f}}_0+{\textbf{f}}_1,{\textbf{g}}_0+{\textbf{g}}_1)$ lies in the lattice $L_{{\textbf{h}}_0+{\textbf{h}}_1}$, and the vector $({\textbf{f}}_0-{\textbf{f}}_1,{\textbf{g}}_0-{\textbf{g}}_1)$ lies in the lattice $L_{{\textbf{h}}_0-{\textbf{h}}_1}.$ The rest follows from Theorem \ref{matrixofdihedral} that for $-N+1\leq r\leq N-1,$ the vectors $(\normalfont{{\textbf{f}}_0^{(r)},{\textbf{f}}_1^{(-r)},{\textbf{g}}_0^{(r)},{\textbf{g}}_1^{(-r)}}),$ $ (\normalfont{{\textbf{f}}_1^{(r)},{\textbf{f}}_0^{(-r)},{\textbf{g}}_1^{(r)},{\textbf{g}}_0^{(-r)}})$ lie in the lattice $L_{\normalfont{{\textbf{h}}}}$.
\end{proof}
\begin{theorem}[Pull-back]\label{pullback}
  Let vectors $\normalfont({\textbf{f}}_0^{\prime},{\textbf{g}}_0^{\prime})$ and $\normalfont({\textbf{f}}_1^{\prime},{\textbf{g}}_1^{\prime})$ lie in the lattices $\normalfont L_{{\textbf{h}}_0+{\textbf{h}}_1}$ and $\normalfont L_{{\textbf{h}}_0-{\textbf{h}}_1}$, respectively. Then, the vector $\normalfont({\textbf{f}}_0^{\prime}+{\textbf{f}}_1^{\prime},{\textbf{f}}_0^{\prime}-{\textbf{f}}_1^{\prime},{\textbf{g}}_0^{\prime}+{\textbf{g}}_1^{\prime},{\textbf{g}}_0^{\prime}-{\textbf{g}}_1^{\prime})$ lies in the lattice $L_{\normalfont{\textbf{h}}}$.
\end{theorem}
\begin{proof}
Let the vectors ${\textbf{u}}_0^{\prime}$ and ${\textbf{u}}_1^{\prime}$ be such that 
\begin{align*}
({\textbf{f}}_0^{\prime},{\textbf{u}}_0^{\prime})\begin{pmatrix}
        I_{N} & H_0+H_1\\
        \00_N & qI_N
    \end{pmatrix} &= ({\textbf{f}}_0^{\prime},{\textbf{g}}_0^{\prime})\\
({\textbf{f}}_1^{\prime},{\textbf{u}}_1^{\prime})\begin{pmatrix}
        I_{N} & H_0-H_1\\
        \00_N & qI_N
    \end{pmatrix} &= ({\textbf{f}}_1^{\prime},{\textbf{g}}_1^{\prime}).
\end{align*}
Therefore we get
\begin{align*}
      {\textbf{f}}_0^{\prime}(H_0+H_1)+q{\textbf{u}}_0^{\prime} &= {\textbf{g}}_0^{\prime}\\
{\textbf{f}}_1^{\prime}(H_0-H_1)+q{\textbf{u}}_1^{\prime} &= {\textbf{g}}_1^{\prime}.
 \end{align*}
adding and subtracting these equations gives    
 \begin{align*}({\textbf{f}}_0^{\prime}+{\textbf{f}}_1^{\prime})H_0+({\textbf{f}}_0^{\prime}-{\textbf{f}}_1^{\prime})H_1+q({\textbf{u}}_0^{\prime}+{\textbf{u}}_1^{\prime}) &= {\textbf{g}}_0^{\prime}+{\textbf{g}}_1^{\prime}\\
 ({\textbf{f}}_0^{\prime}+{\textbf{f}}_1^{\prime})H_1+({\textbf{f}}_0^{\prime}-{\textbf{f}}_1^{\prime})H_0+q({\textbf{u}}_0^{\prime}-{\textbf{u}}_1^{\prime}) &={\textbf{g}}_0^{\prime}-{\textbf{g}}_1^{\prime}.
\end{align*}
Finally
\begin{equation*}
({\textbf{f}}_0^{\prime}+{\textbf{f}}_1^{\prime},{\textbf{f}}_0^{\prime}-{\textbf{f}}_1^{\prime},{\textbf{u}}_0^{\prime}+{\textbf{u}}_1^{\prime},{\textbf{u}}_0^{\prime}-{\textbf{u}}_1^{\prime}) M_{\normalfont{\textbf{h}}} = ({\textbf{f}}_0^{\prime}+{\textbf{f}}_1^{\prime},{\textbf{f}}_0^{\prime}-{\textbf{f}}_1^{\prime},{\textbf{g}}_0^{\prime}+{\textbf{g}}_1^{\prime},{\textbf{g}}_0^{\prime}-{\textbf{g}}_1^{\prime})
\end{equation*}
where 
\begin{equation*}
    M_{\normalfont{\textbf{h}}} = \left(
\begin{array}{c | c}
\begin{matrix}
I_N & \00_N\\
\00_N & I_N
\end{matrix} & \begin{matrix}
H_0 & H_1\\
H_1 & H_0
\end{matrix} \\
\hline
\begin{matrix}
0_N & 0_N\\
0_N & 0_N
\end{matrix} & \begin{matrix}
qI_N & 0_N\\
0_N & qI_N
\end{matrix}
\end{array}
\right)
.\end{equation*}
This gives that the vector $({\textbf{f}}_0^{\prime}+{\textbf{f}}_1^{\prime},{\textbf{f}}_0^{\prime}-{\textbf{f}}_1^{\prime},{\textbf{g}}_0^{\prime}+{\textbf{g}}_1^{\prime},{\textbf{g}}_0^{\prime}-{\textbf{g}}_1^{\prime})$ lies in the lattice $L_{{\textbf{h}}}$. We say that the vector $({\textbf{f}}_0^{\prime}+{\textbf{f}}_1^{\prime},{\textbf{f}}_0^{\prime}-{\textbf{f}}_1^{\prime},{\textbf{g}}_0^{\prime}+{\textbf{g}}_1^{\prime},{\textbf{g}}_0^{\prime}-{\textbf{g}}_1^{\prime})$ is the pull-back of the vectors $({\textbf{f}}_0^{\prime},{\textbf{g}}_0^{\prime})$ and $({\textbf{f}}_1^{\prime},{\textbf{g}}_1^{\prime})$ lying in the lattices $L_{{\textbf{h}}_0+{\textbf{h}}_1}$ and $L_{{\textbf{h}}_0-{\textbf{h}}_1}$, respectively.
\end{proof}
\subsection{Recovering a decryption key}
Let ${\textbf{f}} = ({\textbf{f}}_0,{\textbf{f}}_1), ~~{\textbf{g}} = ({\textbf{g}}_0,{\textbf{g}}_1)$ be ternary vectors where ${\textbf{f}}_i,{\textbf{g}}_i$ roughly have $\lfloor\frac{N}{3}\rfloor$ number of $1$s, $\lfloor\frac{N}{3}\rfloor$ number of $-1$s, and rest are $0$s, the same holds true for ${\textbf{f}}_i^{(r)},{\textbf{g}}_i^{(r)}$, where $-N+1\leq r\leq N-1$. Let $({\textbf{f}}, {\textbf{g}})$ be the private key and ${\textbf{h}} = ({\textbf{h}}_0,{\textbf{h}}_1)$ be the public key satisfying ${\textbf{f}}\star {\textbf{h}} = {\textbf{g}} \pmod q$. From Theorem \ref{pullback}, we know that the vector $(\normalfont{\textbf{f}}_0^{(r)}+{\textbf{f}}_1^{(-r)},{\textbf{g}}_0^{(r)}+{\textbf{g}}_1^{(-r)})\in L_{{\textbf{h}}_0+{\textbf{h}}_1}$, and the vectors $(\normalfont{\textbf{f}}_0^{(r)}-{\textbf{f}}_1^{(-r)},{\textbf{g}}_0^{(r)}-{\textbf{g}}_1^{(-r)}),(\normalfont{\textbf{f}}_1^{(r)}-{\textbf{f}}_0^{(-r)},{\textbf{g}}_1^{(r)}-{\textbf{g}}_0^{(-r)})\in L_{{\textbf{h}}_0-{\textbf{h}}_1}$. In the extreme case when all the $1$s and $-1$s of ${\textbf{f}}_0^{(r)}$ match with $1$s and $-1$s of ${\textbf{f}}_1^{(-r)}$, respectively, and the same is true for ${\textbf{g}}_0^{(r)}, {\textbf{g}}_1^{(-r)}$, then $\norm{({\textbf{f}}_0^{(r)}+{\textbf{f}}_1^{(-r)},{\textbf{g}}_0^{(r)}+{\textbf{g}}_1^{(-r)})}\approx \sqrt{2}\sqrt{\frac{8N}{3}}$. Similarly, we have  $\norm{(\normalfont{\textbf{f}}_0^{(r)}-{\textbf{f}}_1^{(-r)},{\textbf{g}}_0^{(r)}-{\textbf{g}}_1^{(-r)})},\norm{(\normalfont{\textbf{f}}_1^{(r)}-{\textbf{f}}_0^{(-r)},{\textbf{g}}_1^{(r)}-{\textbf{g}}_0^{(-r)})}\approx \sqrt{2}\sqrt{\frac{8N}{3}}$.
Gaussian heuristic predicts that the length of shortest vector in the lattices $L_{{\textbf{h}}_0+{\textbf{h}}_1}$ and $L_{{\textbf{h}}_0-{\textbf{h}}_1}$ is
$$\sigma(L_{{\textbf{h}}_0+{\textbf{h}}_1})=\sigma(L_{{\textbf{h}}_0-{\textbf{h}}_1}) =  \sqrt{\frac{qN}{\pi e}}\approx \sqrt{\frac{4}{\pi e}}N\approx 0.68N,~~~~\text{since}~~~q\approx 4N.$$
Also the ratios,\\

\noindent
$\frac{\norm{({\textbf{f}}_0^{(r)}+{\textbf{f}}_1^{(-r)},{\textbf{g}}_0^{(r)}+{\textbf{g}}_1^{(-r)})}}{\sigma(L_{{\textbf{h}}_0+{\textbf{h}}_1})} = \frac{\norm{({\textbf{f}}_0^{(r)}-{\textbf{f}}_1^{(-r)},{\textbf{g}}_0^{(r)}-{\textbf{g}}_1^{(-r)})}}{\sigma(L_{{\textbf{h}}_0-{\textbf{h}}_1})} = \frac{\norm{({\textbf{f}}_1^{(r)}-{\textbf{f}}_0^{(-r)},{\textbf{g}}_1^{(r)}-{\textbf{g}}_0^{(-r)})}}{\sigma(L_{{\textbf{h}}_0-{\textbf{h}}_1})} \approx \frac{3.37}{\sqrt{N}}.$\\ 

Therefore, for large values of $N$, we expect that for $-N+1\leq r\leq N-1$, the vectors $({\textbf{f}}_0^{(r)}+{\textbf{f}}_1^{(-r)},{\textbf{g}}_0^{(r)}+{\textbf{g}}_1^{(-r)})$ are the shortest in the lattice  $L_{{\textbf{h}}_0+{\textbf{h}}_1}$, and the vectors $({\textbf{f}}_0^{(r)}-{\textbf{f}}_1^{(-r)},{\textbf{g}}_0^{(r)}-{\textbf{g}}_1^{(-r)}), ({\textbf{f}}_1^{(r)}-{\textbf{f}}_0^{(-r)},{\textbf{g}}_1^{(r)}-{\textbf{g}}_0^{(-r)})$ are the shortest in the lattice $L_{{\textbf{h}}_0-{\textbf{h}}_1}$. 

Suppose applying the basis reduction algorithms\footnote{To avoid any possible confusion, in this paper the term \textit{lattice reduction} refers to our reduction method given in subsection \ref{our-reduction}, while the \textit{basis reduction/lattice basis reduction} refers to applying a reduction algorithm like LLL or BKZ to reduce the basis of a lattice.} return $({\textbf{f}}_0^{(r)}+{\textbf{f}}_1^{(-r)},{\textbf{g}}_0^{(r)}+{\textbf{g}}_1^{(-r)})\in L_{{\textbf{h}}_0+{\textbf{h}}_1}$ and  $({\textbf{f}}_0^{(s)}-{\textbf{f}}_1^{(-s)},{\textbf{g}}_0^{(s)}-{\textbf{g}}_1^{(-s)})$ or  $({\textbf{f}}_1^{(t)}-{\textbf{f}}_0^{(-t)},{\textbf{g}}_1^{(t)}-{\textbf{g}}_0^{(-t)})\in L_{{\textbf{h}}_0-{\textbf{h}}_1}$, as a solution to the SVP. In case $s=r$ or $t=-r$, we can obtain the vector $({\textbf{f}}_0^{(s)},{\textbf{f}}_1^{(-s)},{\textbf{g}}_0^{(s)},{\textbf{g}}_1^{(-s)})$ or $({\textbf{f}}_0^{(-t)},{\textbf{f}}_1^{(t)},{\textbf{g}}_0^{(-t)},{\textbf{g}}_1^{(t)})$ that lies in the lattice $L_{{\textbf{h}}}$. For small values of $N$, the chances of getting a match $s=r$ or $t=-r$ is low. However, we observed experimentally that the chance of getting the desired match increases with the value of $N$. Hence, we are able to recover the ternary vectors in the lattice $L_{{\textbf{h}}}$ with high probability for large values of $N$, and by Remark \ref{remark1} it will also work as a decryption key.

As discussed, there are chances that the basis reduction algorithms do not return the desired vectors. Even in those cases there is a way to get a pair of vectors ${\textbf{f}}^{\prime},{\textbf{g}}^{\prime}$ with small coefficients such that $({\textbf{f}}^{\prime},{\textbf{g}}^{\prime})\in L_{{\textbf{h}}}$. Again by Remark \ref{remark1}, such a pair can serve as a decryption key. Suppose $({\textbf{f}}_0^{\prime},{\textbf{g}}_0^{\prime})\in L_{{\textbf{h}}_0+{\textbf{h}}_1}$ and $({\textbf{f}}_1^{\prime},{\textbf{g}}_1^{\prime})\in L_{{\textbf{h}}_0-{\textbf{h}}_1}$ are the short vectors returned by the basis reduction algorithms. Assuming that these vectors take values from the set $\{0,\pm 1, \pm 2, \ldots, \pm \ell\}$. Then, from Theorem \ref{pullback} we know that the vector $({\textbf{f}}_0^{\prime}+{\textbf{f}}_1^{\prime},{\textbf{f}}_0^{\prime}-{\textbf{f}}_1^{\prime},{\textbf{g}}_0^{\prime}+{\textbf{g}}_1^{\prime},{\textbf{g}}_0^{\prime}-{\textbf{g}}_1^{\prime})$ lies in the lattice $L_{{\textbf{h}}}$, and takes values from the set $\{0,\pm 1, \pm 2, \ldots, \pm 2\ell\}$. Let ${\textbf{f}}^{\prime} = ({\textbf{f}}_0^{\prime}+{\textbf{f}}_1^{\prime},{\textbf{f}}_0^{\prime}-{\textbf{f}}_1^{\prime})$ and ${\textbf{g}}^{\prime} = ({\textbf{g}}_0^{\prime}+{\textbf{g}}_1^{\prime},{\textbf{g}}_0^{\prime}-{\textbf{g}}_1^{\prime})$, then $({\textbf{f}}^{\prime},{\textbf{g}}^{\prime})$ serves as a potential decryption key for small values of $\ell$, refer to Table \ref{table: pull back}.
 
\begin{remark}\label{remark2}
The same approach can be applied to recover a decryption key in GR-NTRU built over other groups whose matrices show similar pattern as the dihedral group. For example, a cyclic  group $G = C_{2N}$ of order $2N$ has the matrix of the form $\begin{pmatrix}
A & B \\
B & A
\end{pmatrix}$.
\end{remark}
\section{Experimental Results}
\label{results}
We ran our experiment using different parameter sets $(N, p, q, d)$ where $N$ is a prime number equal to half the order of the dihedral group, $p=3$, $d = \lfloor \frac{2N}{3} \rfloor$, and $q$ is the least power of $2$ satisfying $q>(6d+1)p$.
For each parameter set, we generated $100$ random private keys and messages, then we generated the corresponding public keys and the ciphertexts.
We ran algorithms ~\ref{alg:naive_approach}, ~\ref{alg:ternary key} to retrieve a decryption key and decrypt the ciphertext.
Algorithm ~\ref{alg:naive_approach} describes the steps of retrieving the private key by solving the SVP in a $4N$-dimensional lattice, which is the naive way to do a lattice attack on the public key for an NTRU-like scheme, while Algorithm ~\ref{alg:ternary key} shows the steps of retrieving the key by solving the SVP in two lattices of dimension $2N$.

We measured the average time to run the algorithms, the percentage of the returned vectors that worked successfully as decryption keys, and their average norms.
For the naive approach, if Algorithm ~\ref{alg:naive_approach} returns a vector ${\textbf{k}}$ such that ${\norm{\textbf{k}}} \leq 4 \times \norm{\textbf{private key}}$, we count the trail as a success. For the pull-back approach, the algorithm returns two vectors ${\textbf{k}_1}$ (non-ternary) and ${\textbf{k}_2}$ (ternary). In case ~ $~{\norm{\textbf{k}_1}} \leq 4 \times \norm{\textbf{private key}}$, then ${\textbf{k}_1}$ is counted as a decryption key, while ${\textbf{k}_2}$ being ternary is always counted, if the algorithm returns it.
For verification purpose, we have decrypted the ciphertext and checked that the decrypted messages equal the original messages for all the returned keys. 

The success of retrieving the key highly depends on the basis reduction algorithm. The goal of basis reduction techniques is to find shorter and nearly orthogonal bases. LLL ~\citep{LLL} and BKZ ~\citep{BKZ} are famous examples of these algorithms. While LLL can run in a polynomial time and produce a good reduced basis for smaller dimensions, it fails in higher dimensions. BKZ has an additional input, the block size $\beta$, that affects both the running time and the quality of the reduced basis. The larger the value of $\beta$, the better the quality of the reduced basis and the higher the running time. Many enhancements have been introduced to BKZ resulting in BKZ2.0 ~\citep{BKZ2.0} and other variants of BKZ ~\citep{progressive_bkz, deep_bkz}.
For our experiment, FPLLL implementations of LLL and BKZ2.0 ~\citep{fplll} have been used as options for lattice basis reduction. We executed the experiment in SageMath depending on FPyLLL ~\citep{fpylll} as a python wrapper of FPLLL. Tables ~\ref{table:naive approach}, ~\ref{table: pull back} show the results for algorithms ~\ref{alg:naive_approach}, ~\ref{alg:ternary key}, respectively and Fig. ~\ref{fig:main} breaks down these results into comparisons of the success rates and average times for the mentioned algorithms.
Timed results have been executed on a system running Windows 10 Pro with Intel(R) Core(TM)i9-10980HK CPU@2.40GHZ with 32 GB installed RAM. 

We can notice that the naive approach is successful up-to certain values of $N$, then for $N> 67$ we aren't able to get the results in reasonable time since we are solving the SVP in a $4N$-dimensional lattice. However, the pull-back approach can function and retrieve the decryption key for larger values of $N$.\\
 In the pull-back approach, the success rate of getting the non-ternary key ${\textbf{k}_1}$ is higher than that for the ternary key ${\textbf{k}_2}$ since, in the former case, we are looking only for two short-enough vectors in the two lattices $L_{{\textbf{h}}_0+{\textbf{h}}_1}$, $L_{{\textbf{h}}_0-{\textbf{h}}_1}$ such that their pull-back is also a short-enough vector in the larger lattice $L_{{\textbf{h}}}$. However, to retrieve ${\textbf{k}_2}$, we need to find a match for two rotated vectors $\in \{ {0,\pm 1, \pm 2}\}^{2N} $ in the lattices  $L_{{\textbf{h}}_0+{\textbf{h}}_1}$, $L_{{\textbf{h}}_0-{\textbf{h}}_1}$ that enable retrieving a ternary key in $L_{{\textbf{h}}}$. Finding such a match increases with the dimension of the lattice and the quality of vectors returned by the basis reduction algorithm. We can see from Table ~\ref{table: pull back}, Fig. \ref{main:e} that LLL is successful up-to $N=73$ in finding the non-ternary key, then the norm of the returned vector starts to increase significantly, while BKZ2.0 can retrieve vectors with smaller norms for higher dimensions.
On the other hand, the success rate for finding a ternary vector starts to increase when the value of  $N$ increases. Further, for $N = 61, 67, 71$, and $73$ the success rate is $98\%, 73\%, 48\%$, and $33\%$ respectively for LLL as a basis reduction option, while the success percentage is $100\%$ for the same values of $N$ with  BKZ2.0.

We would like to point out that the effect that appears as a sudden drop in the success rate in Fig. ~\ref{main:b} doesn't mean that BKZ2.0 fails to find a decryption key for $N>67$ in the naive approach and for $N>131$ in the pull-back approach, rather it fails to find the key in a reasonable time due to enabling auto-abort flag in our experiment.  \\


\RestyleAlgo{ruled}
\SetKwComment{Comment}{/* }{ */}
\SetKwInput{KwInput}{Input}                
\SetKwInput{KwOutput}{Output} 

 \begin{algorithm}[]
\DontPrintSemicolon

\caption{Naive approach to retrieve a decryption key}\label{alg:naive_approach}
\KwInput{$\textbf{N, p, q}$: parameters of the dihedral group based GR-NTRU  \newline
            ${\textbf{h}}$: the public key \newline 
            \textbf{threshold}: the maximum norm of a vector to check as a key\newline
            \textbf{option} : basis reduction algorithm
}
\KwOutput{${\textbf{k}}\in L_{{\textbf{h}}}$ that serves as a decryption key or a failure }
\vspace{0.3cm}
$M_{{\textbf{h}} } \gets$  \FuncSty{get\_lattice\_basis(h,N,q)} \Comment*[r]{4N $\times$ 4N matrix}
\vspace{0.2cm}
${M_{{\textbf{h}}}^ {Reduced} } \gets$  \FuncSty{reduce\_basis($M_{{\textbf{h}} }$,option)} \;
$i \gets 1$\;
\While{$i \leq 4N$}{
  Let ${\textbf{v}}= (v_1, v_2,  \dots v_{4N})$ be the $i^{th}$ shortest vector of ${M_{{\textbf{h}}}^ {Reduced} }$\;
\If{$\norm{{\normalfont \textbf{v}}} > threshold$} 
{\Return{failure}}
Let ${\textbf{f}}^{\prime} = (v_1, v_2, \dots v_{2N})$, ${\textbf{g}}^{\prime} = (v_{2N+1}, v_{2N+1}, \dots v_{4N})$\;
\If{${\normalfont \textbf{f}}^{\prime}$ is invertible in $\Z_p D_N$ }
{
${\textbf{k}} \gets (\normalfont{\textbf{f}}^{\prime},{\textbf{g}}^{\prime})$\;
\Return{$\normalfont{\textbf{k}}$}}
$i \gets i+1 $\;
}

\end{algorithm}
\begin{algorithm}[tbp]
\DontPrintSemicolon
\caption{Pull-back approach to retrieve decryption key}\label{alg:ternary key}
\KwInput{$\textbf{N, p, q}$: parameters of the dihedral group based GR-NTRU \newline
            ${\textbf{h}} = ({\textbf{h}}_0, {\textbf{h}}_1)$: the public key \newline 
            \textbf{threshold}: the maximum norm of a vector to check as a key \newline
            \textbf{option} : basis reduction algorithm
}
\KwOutput{${\textbf{k}_1}\in L_{{\textbf{h}}}$ that serves as a decryption key, 
${\textbf{k}_2}\in L_{{\textbf{h}}}$ a \textbf{ternary} decryption key, or a failure }

${\textbf{k}_1} \gets [ \hspace{0.1cm} ]$, ${\textbf{k}_2} \gets [ \hspace{0.1cm} ]$\; $ \normalfont{ key_1\_found} \gets \textbf{false} $, $key_2\_found \gets \textbf{false} $\Comment*[r]{initialization}
$M_{{\textbf{h}}_0+{\textbf{h}}_1} \gets$  \FuncSty{get\_first\_basis(h,N,q)} \Comment*[r]{2N $\times$ 2N matrix}
$M_{{\textbf{h}}_0-{\textbf{h}}_1} \gets$  \FuncSty{get\_second\_basis(h,N,q)} \Comment*[r]{2N $\times$ 2N matrix}
$M_{{\textbf{h}}_0+{\textbf{h}}_1}^{Reduced}, M_{{\textbf{h}}_0-{\textbf{h}}_1}^{Reduced} \gets$  \FuncSty{reduce\_basis($M_{{\textbf{h}}_0+{\textbf{h}}_1}$,$M_{{\textbf{h}}_0-{\textbf{h}}_1}$,option)} \;
$i \gets 1$\;
\While{$i \leq 2N$}{
Let ${\textbf{v}}= (v_1, v_2,  \dots v_{2N})$ be the $i^{th}$ shortest vector of $M_{{\textbf{h}}_0+{\textbf{h}}_1}^{Reduced}$\;
\If{$\norm{{\normalfont \textbf{v}}} > threshold$}
{
\If{$key_1\_found $ \textbf{or}  $key_2\_found$}
{
\Return{${\textbf{k}_1}, {\textbf{k}_2}$}
}
\Return{failure}

}
Let ${\textbf{f}}_0^{\prime} = (v_1, v_2, \dots v_{N})$, ${\textbf{g}}_0^{\prime} = (v_{N+1}, v_{N+2}, \dots v_{2N})$\;
$j \gets 1$\;
\While{$j \leq 2N$}{
Let ${\textbf{w}}= (w_1, w_2,  \dots w_{2N})$ be the $j^{th}$ shortest vector of $M_{{\textbf{h}}_0-{\textbf{h}}_1}^{Reduced}$\;
\If{$\norm{{\normalfont \textbf{w}}} > threshold$}
{break}
Let ${\textbf{f}}_1^{\prime} = (w_1, w_2, \dots w_{N})$, ${\textbf{g}}_1^{\prime} = (w_{N+1}, w_{N+2}, \dots w_{2N})$\;

$(\normalfont{\textbf{f}}^{\prime},{\textbf{g}}^{\prime}) \gets$ 
$\left(({\textbf{f}}_0^{\prime}+{\textbf{f}}_1^{\prime},{\textbf{f}}_0^{\prime}-{\textbf{f}}_1^{\prime}),({\textbf{g}}_0^{\prime}+{\textbf{g}}_1^{\prime},{\textbf{g}}_0^{\prime}-{\textbf{g}}_1^{\prime})\right)$

\If{\textbf{not}($key_1\_found$) \textbf{and} ${\normalfont \textbf{f}}^\prime$ is invertible in $\Z_p D_N$ }
        {
        $\normalfont{\textbf{k}}_1 \gets (\normalfont{\textbf{f}}^{\prime},{\textbf{g}}^{\prime}) $\;
        $key_1\_found \gets \textbf{true}$ \;
        
        }
        
\If{\textbf{not}($key_2\_found$) \textbf{and} $(\normalfont{\textbf{f}}^{\prime},{\textbf{g}}^{\prime}) \in \{-2, 0, 2 \}^{4N}$}
{
  $(\normalfont{\textbf{f}}^{\prime\prime},{\textbf{g}}^{\prime\prime}) \gets (\frac{\normalfont{\textbf{f}}^{\prime}}{2},\frac{{\textbf{g}}^{\prime}}{2}) $ \Comment*[r]{coefficient-wise division}
  \If{${\normalfont \textbf{f}}^{\prime \prime}$ is invertible in $\Z_p D_N$ }
        {
        $ \normalfont{\textbf{k}}_2 \gets (\normalfont{\textbf{f}}^{\prime\prime},{\textbf{g}}^{\prime\prime}) $\;
          $key_2\_found \gets \textbf{true}$  \;
        
        }

}
\If{$key_1\_found $ \textbf{and}  $key_2\_found$}
{
\Return{$\normalfont{\textbf{k}_1}, \normalfont{\textbf{k}_2} $}
}
$j \gets j+1 $\;
}
$i \gets i+1 $\;
}

\end{algorithm}

\renewcommand{\thefootnote}{\fnsymbol{footnote}}

\begin{table}[]
\begin{center}
\begin{minipage}{\textwidth}
\setcounter{mpfootnote}{-1}%
\renewcommand{\thempfootnote}{\fnsymbol{mpfootnote}}
\renewcommand\footnoterule{}
\caption{Results for the naive approach to retrieve a decryption key}\label{table:naive approach}
\begin{NoHyper}
 \begin{tabular*}{\textwidth}{@{\extracolsep{\fill}}llcccc@{\extracolsep{\fill}}}
          \toprule%
              \multicolumn{2}{c}\textbf{ {}} &\multicolumn{2}{c}{\textbf{LLL }} &\multicolumn{2}{c}{\textbf{BKZ2.0 }} \\
          \midrule
          N & $\norm{key}$ & ${\textbf{k}} \%$ & \makecell{Time \\ avg (s)}   & ${\textbf{k}} \%$ &   \makecell{Time \\ avg (s)}\\
          \midrule
               13 & 5.7446 &100 & 0.3219&	100 &	0.872\\
               17 & 6.7082 &100 &0.5431 & 100 & 1.118\\
               19 & $ 7 $ &100 & 0.7139 & 100 & 1.337 \\
               23& 7.8103 & 100 & 1.2061 & 100 & 	2.182\\
               29& 8.7750& 100& 2.5117 &100 &	5.057\\
               31& 9 & 100  & 3.1766 & 100 &6.189	\\
               37& 9.8489 & 100 & 6.4258 & 100 & 12.038\\
               41& 10.4403 & 7 & 9.1183 & 100 & 21.484 \\
               43& 10.6301 & 1 & 16.640 & 100 & 33.765\\
               47& 11.1803 & 1 & 37.187 & 100 & 42.984\\
              
               53& 11.8743 & 0 & \_ & 100 &  1823.4\footnotemark[1]\\ 
               61& 12.6886 & 0 & \_ & 100 & 4203.9\footnotemark[1]\\
               67& 13.3041 & 0 & \_ & 100 & 9885.4\footnotemark[1]\\
                \makecell{$ 71$   } & \makecell{ 13.7477 } & \makecell{0   }  & \makecell{\_  } & \makecell{ \_ \footnotemark[7]\\   } &\makecell{ \_  }\\

          \bottomrule
          \label{tab:reproduzierbarkeit}
      \end{tabular*}
  \end{NoHyper}    
\stepcounter{mpfootnote}%
\footnotetext{
The results are obtained by running Algorithm ~\ref{alg:naive_approach} (the naive approach) with threshold value = $4 \times \norm{ key}$, where
$\norm{key}$ refers to the norm of the private key for the corresponding value of $N$, ${\textbf{k}} \%$ refers to the success rate of retrieving a decryption key, which is equivalent to solve the SVP in a $4N$-dimensional lattice, and the average time indicates the average running time of one trail of Algorithm ~\ref{alg:naive_approach} over $100$ randomly generated examples.
LLL has been called using the default parameters with $\delta =0.99, \eta=0.501$ and BKZ2.0 with block size $\beta = 40$ and auto-abort flag enabled.

}

\footnotetext[*]{ We have used multiple-precision binary floating-point computations with correct rounding (MPFR) ~\citep{mpfr} for arbitrary precision at 200-bits   due to floating points errors. 
}
\footnotetext[**]{ auto-abort is triggered when the execution takes longer times and the quality of basis doesn't improve quickly over tours.  }

\end{minipage}
\end{center}
\end{table}

\begin{table}[] 
\begin{center}
\begin{minipage}{\textwidth}
\setcounter{mpfootnote}{-1}%
\renewcommand{\thempfootnote}{\fnsymbol{mpfootnote}}

\renewcommand\footnoterule{}
\caption{Results for the pull back approach to retrieve a decryption key}\label{table: pull back}
\begin{NoHyper}
    
 \begin{tabular*}{\textwidth}{@{\extracolsep{\fill}}llcccccccc@{\extracolsep{\fill}}}
          \toprule%
              \multicolumn{2}{c}\textbf{ {}} &\multicolumn{4}{c}{\textbf{LLL }} &\multicolumn{4}{c}{\textbf{BKZ2.0  }} \\
          \midrule
          N & $\norm{key}$  & ${\textbf{k}_1} \%$ &  \makecell{ $ {\norm{\textbf{k}_1}}$ \\ avg }& ${\textbf{k}_2} \%$ & \makecell{Time \\ avg (s)}   & ${\textbf{k}_1} \%$ & \makecell{$ {\norm{\textbf{k}_1}}$ \\ avg }& ${\textbf{k}_2} \%$ &\makecell{Time \\ avg (s)} \\
          \midrule
               13 & 5.745 & 100 & 9.180 &  42& 0.427 &100 &9.194	&57& 0.998 \\
               17 & 6.708 & 100 & 10.947 & 84 &0.748 &100&10.993&78&1.439 \\
               19 & 7 & 100 & 11.561 &86& 1.024  &100&11.648&88&1.713 \\
               23& 7.810 & 100 &13.117 &95&1.176&100 &13.212 &96&1.903\\
               29& 8.775 & 100 & 15.003& 99& 1.578& 100&15.008&100&2.380\\
               31& 9 & 100 & 15.387	& 100 &1.844 &100&15.458&100&2.663\\
               37& 9.849 & 100 & 17.057&100&2.296&100&17.084 &100&3.384\\

               41& 10.440 & 100& 18.208& 100 & 3.087 & 100 & 18.224 & 100 & 4.442\\
               
                43& 10.630 & 100 & 18.545 & 100 & 3.963 & 100 & 18.569	& 100 & 5.326\\
                
               47&  11.180 & 100 &19.616 & 100& 5.174 & 100 & 19.603& 100 & 7.815 \\
               53 & 11.874 & 100 & 20.936 & 100 & 7.419 & 100 & 20.912& 100 & 11.04\\
               61& 12.689 & 100 & 22.715 & 98 & 12.749 & 100 & 22.442 & 100 & 21.902\\
                67& 13.304 & 100 & 26.881 & 73 & 26.998 & 100 & 23.750 & 100 & 28.692 \\
                71 & 13.748 & 96 & 33.665 & 48 & 56.955 & 100 & 24.561 & 100 & 60.972\\
                73& 13.892 & 78 & 39.519 & 33 & 58.989& 100 & 24.830& 100 & 77.803\\
                79&  14.457	& 7 & 51.279 & 0 & 37.093\footnotemark[2] & 100 & 25.967 & 100 & 136.291\\
                83 & 14.866 & 0 &  \_ & 0 & \_ & 100 & 26.721& 100 & 160.720 \\
                89& 15.395 & 0 & \_& 0&\_& 100& 27.809 & 100 & 1928.9\footnotemark[1]\\
                97& 16.031& 0& \_& 0 & \_ & 100& 29.189 &100&  2710.7\footnotemark[1]\\
                131 & 18.682 & 0 & \_ & 0 & \_ & 100 & 34.117 & 100 & 16043.8\footnotemark[1]\\
                149 &  19.925 & 0 & \_ & 0 & \_ & \_\footnotemark[7] & \_ & \_\footnotemark[7]  & \_ \\
          \bottomrule
          \label{tab:reproduzierbarkeit}
      \end{tabular*}
      \end{NoHyper}
\stepcounter{mpfootnote}%
\footnotetext{The results are obtained by running Algorithm ~\ref{alg:ternary key} (the pull-back approach) with threshold value = $2 \times \norm{ key}$, where $\norm{key}$ refers to the norm of the private key for the corresponding value of $N$, ${\textbf{k}_1} \%$ refers to the success rate of retrieving a decryption key short-enough with norm ${\norm{\textbf{k}_1}}$, while ${\textbf{k}_2} \%$ refers to the success rate of retrieving a ternary key , which is equivalent to solve two instances of the SVP in a $2N$-dimensional lattice, and the average time indicates the average running time of one trail of Algorithm ~\ref{alg:ternary key} over $100$ randomly generated examples. LLL has been called using the default parameters with $\delta =0.99, \eta=0.501$ and BKZ2.0 with block size $\beta = 40$ and auto-abort flag enabled.}
\footnotetext[*]{We have used multiple-precision binary floating-point computations with correct rounding (MPFR) ~\citep{mpfr} for arbitrary precision at 200-bits   due to floating points errors. 
}
\footnotetext[**]{auto-abort is triggered when the execution takes longer times and the quality of basis doesn't improve quickly over tours. }
\footnotetext[\dagger]{ LLL algorithm produce worse quality of reduced basis, therefore Algorithm ~\ref{alg:ternary key} meets the threshold condition and doesn't do any further processing to find a ternary key ${\textbf{k}_2}$, hence the time is slightly lower for this value of $N$.  }
\end{minipage}
\end{center}
\end{table}


\begin{figure}[!h]

\begin{minipage}{.5\linewidth}
\centering
\subfloat[\centering LLL success\%]{\label{main:a}\includegraphics[scale=.37]{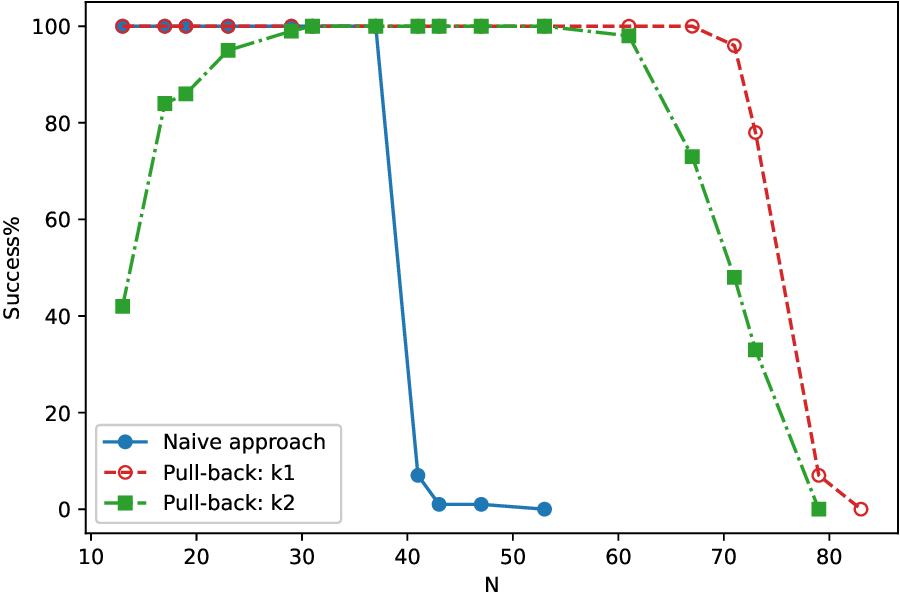}}
\end{minipage}%
\begin{minipage}{.5\linewidth}
\centering
\subfloat[\centering BKZ2.0 success\%]{\label{main:b}\includegraphics[scale=.37]{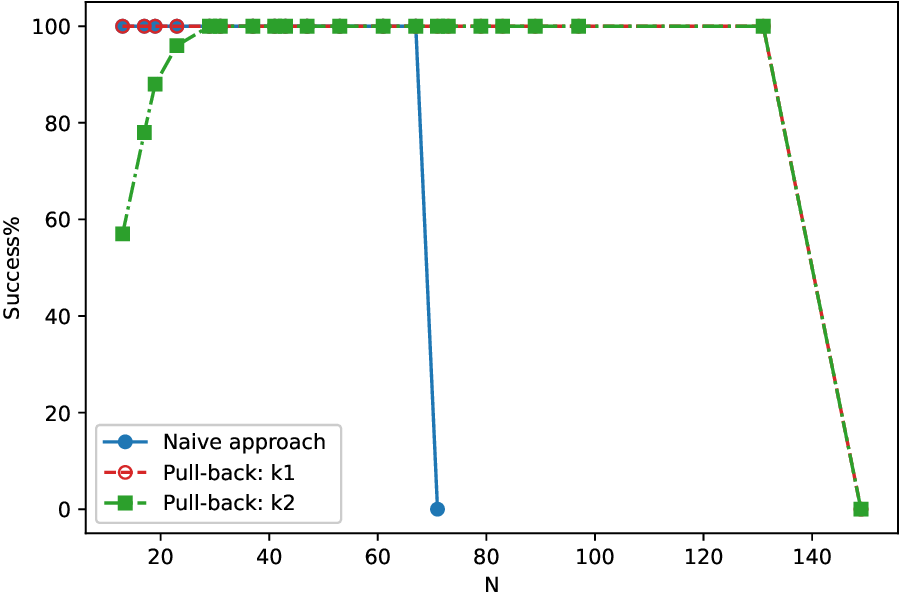}}
\end{minipage}
\par \medskip
\begin{minipage}{.5\linewidth}
\centering
\subfloat[\centering LLL average time (seconds)]{\label{main:c}\includegraphics[scale=.37]{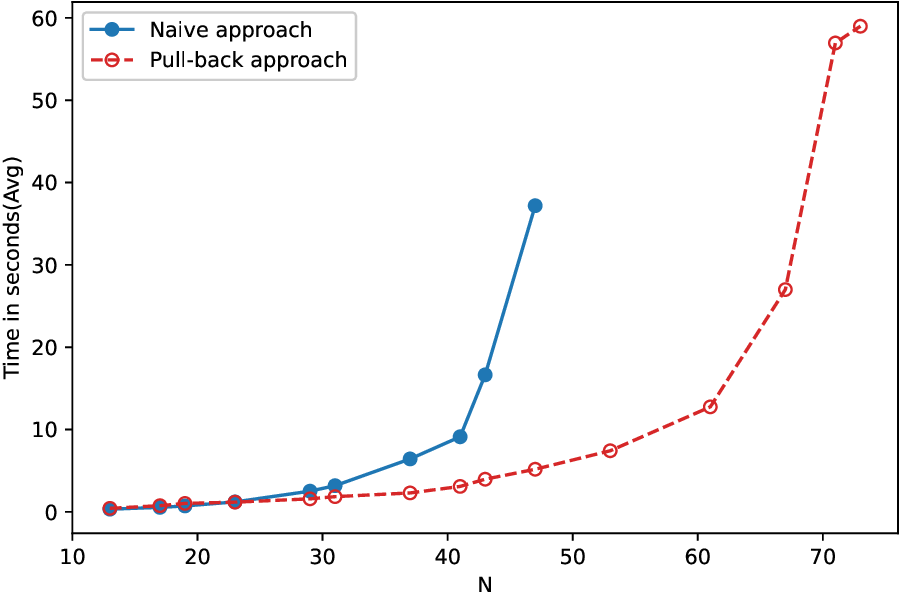}}
\end{minipage}%
\begin{minipage}{.5\linewidth}
\centering
\subfloat[\centering BKZ2.0 average time (seconds)]{\label{main:d}\includegraphics[scale=.37]{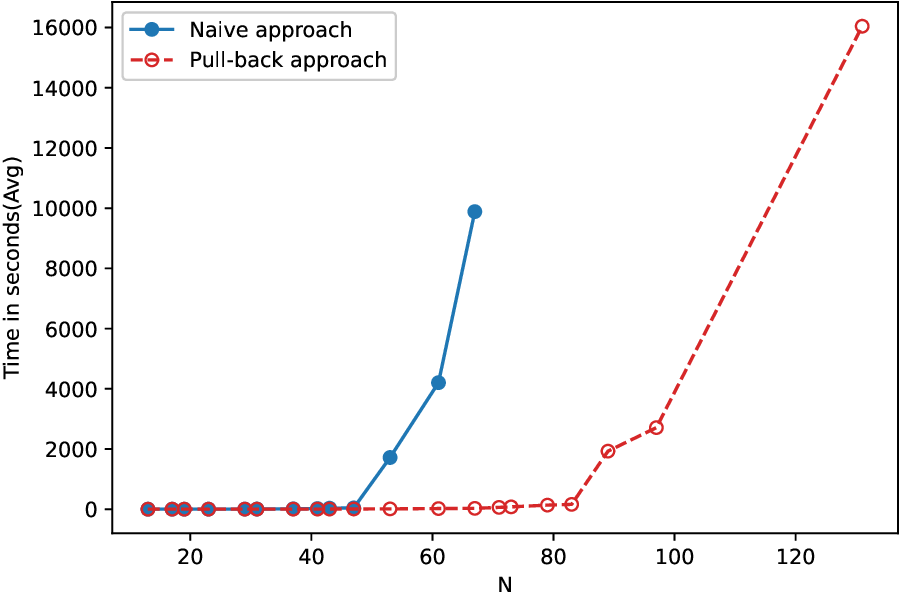}}
\end{minipage}
\par \medskip

\par\medskip
\centering
\subfloat[\centering Pull-back ${\norm{\textbf{k}_1}}$ (LLL vs. BKZ)]{\label{main:e}\includegraphics[scale=.4]{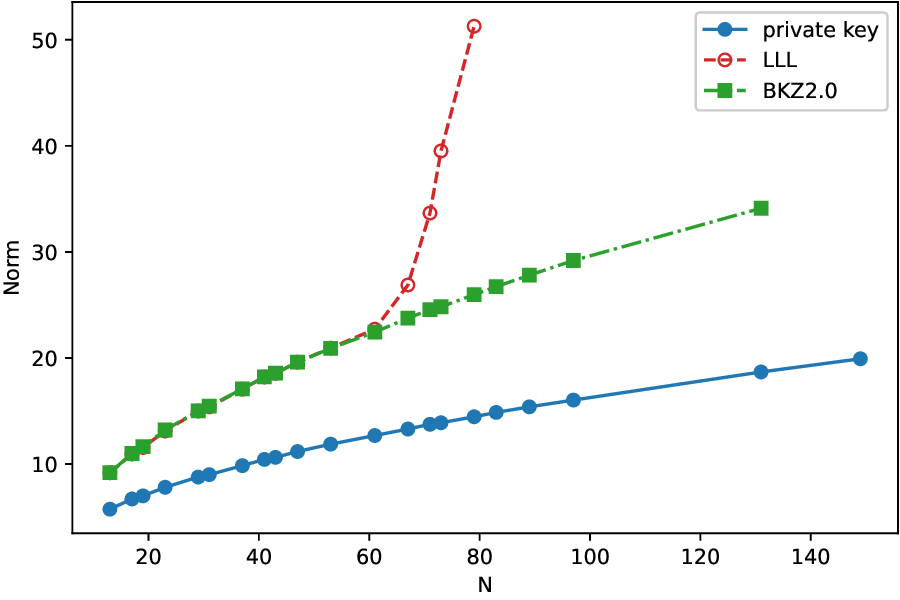}}

\caption{Naive approach vs. Pull-back approach over different values of N. Figures \ref{main:a}, \ref{main:b} show the success percentage of retrieving a decryption key, figures ~\ref{main:c}, \ref{main:d} compare the average time (in seconds) for LLL, BKZ2.0, respectively, and Fig. ~\ref{main:e} compares the norm of the private key with the norm of ${\textbf{k}_1}$ returned by the pull-back approach for LLL and BKZ2.0 over different values of $N$. }
\label{fig:main}
\end{figure}



\clearpage

\section{Conclusion}
\label{conclusion}
This paper provides a lattice reduction for GR-NTRU based on the dihedral group using elementary matrix algebra. We show that for a dihedral group of order $2N$, one can perform the lattice attack on the public key in two $2N$-dimensional lattices instead of a $4N$-dimensional lattice. We provide an approach to retrieve two vectors in the two smaller lattices and pull them back to the larger one. Our pull-back approach gives two potential decryption keys: a short-enough key (not necessarily ternary), while the other is a ternary key. For a good reduced basis, theoretical analysis and experimental results show that retrieving the first key is deterministic, while the ternary key is returned with high probability. 
Thus, the scheme under investigation provides an equivalent level of security to GR-NTRU based on a cyclic group of order $N$.
This study is part of an effort to understand the effect of using nonabelian groups in the context of GR-NTRU. As a future scope, the resistance of other nonabelian groups to lattice attacks should be explored.
\section{Declaration}
\textbf{Competing interest:} The authors have no competing interests to declare that are relevant to the content of this article.

\begin{appendices}
\section{Lattice reduction: toy example}\label{secA1}


For a parameter set as $N=7, p=3, q=128, d= \lfloor \frac{2N}{3} \rfloor = \lfloor \frac{14}{3} \rfloor = 4$.\\
Suppose that the private key $(f,g)$  is sampled as:
\begin{align*}
   f & = x-x^2-x^4+x^5+x^6+y-yx^2+yx^4-yx^6\\
   g & = x-x^2+x^3-x^5-x^6+yx+yx^4-yx^5  
\end{align*}
then the public key $h = f_q\star g  \pmod q$ \\
Therefore, 
\begin{align*}
  h= 115+42x+117x^2+108x^3+73x^4+3x^5+53x^6+29y\\ +108yx+34yx^2+72yx^3+5yx^4+36yx^5+101yx^6   
\end{align*}

In other words, the vectors corresponding to $f, g$ and $ h$ are
 ${\textbf{f}} = ({\textbf{f}}_0,{\textbf{f}}_1)$, ${\textbf{g}} = ({\textbf{g}}_0,{\textbf{g}}_1)$, and ${\textbf{h}} = ({\textbf{h}}_0,{\textbf{h}}_1)$, respectively, where
 \begin{align*}
     {\textbf{f}} &= ({\textbf{f}}_0,{\textbf{f}}_1) =
\left( (0, 1, -1, 0, -1, 1, 1),(1, 0, -1, 0, 1, 0, -1) \right) \\
{\textbf{g}} &= ({\textbf{g}}_0,{\textbf{g}}_1) =
\left( (0, 1, -1, 1, 0, -1, -1),( 0, 1, 0, 0, 1, -1, 0) \right) \\
{\textbf{h}} &= ({\textbf{h}}_0,{\textbf{h}}_1) =
\left( (115, 42, 117, 108, 73, 3, 53),( 29, 108, 34, 72, 5, 36, 101) \right)
 \end{align*}

We can notice that the norm of the private key is : $\norm {({\textbf{f}}, {\textbf{g}})} = \sqrt{17}$.\\
Suppose we want to encrypt the message represented by the vector:
$${\textbf{m}} = (0, 0, -1, -1, -1, 0, 1, -1, -1, 0, 1, -1, 0, 0) $$ then the ciphertext vector  ${\textbf{c}} = p{\textbf{h}}\star {\textbf{r}} + {\textbf{m}} \pmod q \text{ for } {\textbf{r}} = (0, 1, -1, 1, 0, -1,\\ -1, 0, 1, 0, 0, 1, -1, 0)$ will be:

$${\textbf{c}} = (123, 64, 97, 31, 92, 46, 63, 119, 23, 111, 39, 80, 33, 99). $$\\
For an attacker who wants to decrypt the message knowing only the public key ${\textbf{h}}$, he launches the lattice attack on the public key, and for that he has the following options:

\subsection{The naive approach}
The straightforward way to apply the naive attack is given by Algorithm ~\ref{alg:naive_approach}.\\
\textbullet\ Build the matrix $M_{{\textbf{h}} }$ i.e., the basis of the lattice $L_{{\textbf{h}} }$ as the following:
\newcommand\scalemath[2]{\scalebox{#1}{\mbox{\ensuremath{\displaystyle #2}}}}
 \[
\scalemath{0.57}{
M_{{\textbf{h}} }=
   \begin{bmatrix}
  \begin{array}{c|c}
  \begin{array}{c c c c c c c c c c c c c c} 
       1 & 0 & 0 & 0 & 0 & 0 & 0 & 0 & 0 & 0 & 0 & 0 & 0 & 0 \\
0 & 1 & 0 & 0 & 0 & 0 & 0 & 0 & 0 & 0 & 0 & 0 & 0 & 0 \\
0 & 0 & 1 & 0 & 0 & 0 & 0 & 0 & 0 & 0 & 0 & 0 & 0 & 0 \\
0 & 0 & 0 & 1 & 0 & 0 & 0 & 0 & 0 & 0 & 0 & 0 & 0 & 0 \\
0 & 0 & 0 & 0 & 1 & 0 & 0 & 0 & 0 & 0 & 0 & 0 & 0 & 0 \\
0 & 0 & 0 & 0 & 0 & 1 & 0 & 0 & 0 & 0 & 0 & 0 & 0 & 0 \\
0 & 0 & 0 & 0 & 0 & 0 & 1 & 0 & 0 & 0 & 0 & 0 & 0 & 0 \\
0 & 0 & 0 & 0 & 0 & 0 & 0 & 1 & 0 & 0 & 0 & 0 & 0 & 0 \\
0 & 0 & 0 & 0 & 0 & 0 & 0 & 0 & 1 & 0 & 0 & 0 & 0 & 0 \\
0 & 0 & 0 & 0 & 0 & 0 & 0 & 0 & 0 & 1 & 0 & 0 & 0 & 0 \\
0 & 0 & 0 & 0 & 0 & 0 & 0 & 0 & 0 & 0 & 1 & 0 & 0 & 0 \\
0 & 0 & 0 & 0 & 0 & 0 & 0 & 0 & 0 & 0 & 0 & 1 & 0 & 0 \\
0 & 0 & 0 & 0 & 0 & 0 & 0 & 0 & 0 & 0 & 0 & 0 & 1 & 0 \\
0 & 0 & 0 & 0 & 0 & 0 & 0 & 0 & 0 & 0 & 0 & 0 & 0 & 1 \\
  \end{array} 
  &
  \begin{array}{ccccccc :ccccccc} 
115 & 42 & 117 & 108 & 73 & 3 & 53 & 29 & 108 & 34 & 72 & 5 & 36 & 101 \\
53 & 115 & 42 & 117 & 108 & 73 & 3 & 108 & 34 & 72 & 5 & 36 & 101 & 29 \\
3 & 53 & 115 & 42 & 117 & 108 & 73 & 34 & 72 & 5 & 36 & 101 & 29 & 108 \\
73 & 3 & 53 & 115 & 42 & 117 & 108 & 72 & 5 & 36 & 101 & 29 & 108 & 34 \\
108 & 73 & 3 & 53 & 115 & 42 & 117 & 5 & 36 & 101 & 29 & 108 & 34 & 72 \\
117 & 108 & 73 & 3 & 53 & 115 & 42 & 36 & 101 & 29 & 108 & 34 & 72 & 5 \\
42 & 117 & 108 & 73 & 3 & 53 & 115 & 101 & 29 & 108 & 34 & 72 & 5 & 36 \\ \hdashline[2pt/2pt]
29 & 108 & 34 & 72 & 5 & 36 & 101 & 115 & 42 & 117 & 108 & 73 & 3 & 53 \\
108 & 34 & 72 & 5 & 36 & 101 & 29 & 53 & 115 & 42 & 117 & 108 & 73 & 3 \\
34 & 72 & 5 & 36 & 101 & 29 & 108 & 3 & 53 & 115 & 42 & 117 & 108 & 73 \\
72 & 5 & 36 & 101 & 29 & 108 & 34 & 73 & 3 & 53 & 115 & 42 & 117 & 108 \\
5 & 36 & 101 & 29 & 108 & 34 & 72 & 108 & 73 & 3 & 53 & 115 & 42 & 117 \\
36 & 101 & 29 & 108 & 34 & 72 & 5 & 117 & 108 & 73 & 3 & 53 & 115 & 42 \\
101 & 29 & 108 & 34 & 72 & 5 & 36 & 42 & 117 & 108 & 73 & 3 & 53 & 115 \\
  \end{array}\\
 \hline
 \begin{array}{cccccccccccccc} 
      
0 & 0 & 0 & 0 & 0 & 0 & 0 & 0 & 0 & 0 & 0 & 0 & 0 & 0 \\
0 & 0 & 0 & 0 & 0 & 0 & 0 & 0 & 0 & 0 & 0 & 0 & 0 & 0 \\
0 & 0 & 0 & 0 & 0 & 0 & 0 & 0 & 0 & 0 & 0 & 0 & 0 & 0 \\
0 & 0 & 0 & 0 & 0 & 0 & 0 & 0 & 0 & 0 & 0 & 0 & 0 & 0 \\
0 & 0 & 0 & 0 & 0 & 0 & 0 & 0 & 0 & 0 & 0 & 0 & 0 & 0 \\
0 & 0 & 0 & 0 & 0 & 0 & 0 & 0 & 0 & 0 & 0 & 0 & 0 & 0 \\
0 & 0 & 0 & 0 & 0 & 0 & 0 & 0 & 0 & 0 & 0 & 0 & 0 & 0 \\
0 & 0 & 0 & 0 & 0 & 0 & 0 & 0 & 0 & 0 & 0 & 0 & 0 & 0 \\
0 & 0 & 0 & 0 & 0 & 0 & 0 & 0 & 0 & 0 & 0 & 0 & 0 & 0 \\
0 & 0 & 0 & 0 & 0 & 0 & 0 & 0 & 0 & 0 & 0 & 0 & 0 & 0 \\
0 & 0 & 0 & 0 & 0 & 0 & 0 & 0 & 0 & 0 & 0 & 0 & 0 & 0 \\
0 & 0 & 0 & 0 & 0 & 0 & 0 & 0 & 0 & 0 & 0 & 0 & 0 & 0 \\
0 & 0 & 0 & 0 & 0 & 0 & 0 & 0 & 0 & 0 & 0 & 0 & 0 & 0 \\
0 & 0 & 0 & 0 & 0 & 0 & 0 & 0 & 0 & 0 & 0 & 0 & 0 & 0 \\
 \end{array}
&
\begin{array}{cccccccccccccc}
128 & 0 & 0 & 0 & 0 & 0 & 0 & 0 & 0 & 0 & 0 & 0 & 0 & 0 \\
0 & 128 & 0 & 0 & 0 & 0 & 0 & 0 & 0 & 0 & 0 & 0 & 0 & 0 \\
0 & 0 & 128 & 0 & 0 & 0 & 0 & 0 & 0 & 0 & 0 & 0 & 0 & 0 \\
0 & 0 & 0 & 128 & 0 & 0 & 0 & 0 & 0 & 0 & 0 & 0 & 0 & 0 \\
0 & 0 & 0 & 0 & 128 & 0 & 0 & 0 & 0 & 0 & 0 & 0 & 0 & 0 \\
0 & 0 & 0 & 0 & 0 & 128 & 0 & 0 & 0 & 0 & 0 & 0 & 0 & 0 \\
0 & 0 & 0 & 0 & 0 & 0 & 128 & 0 & 0 & 0 & 0 & 0 & 0 & 0 \\
0 & 0 & 0 & 0 & 0 & 0 & 0 & 128 & 0 & 0 & 0 & 0 & 0 & 0 \\
0 & 0 & 0 & 0 & 0 & 0 & 0 & 0 & 128 & 0 & 0 & 0 & 0 & 0 \\
0 & 0 & 0 & 0 & 0 & 0 & 0 & 0 & 0 & 128 & 0 & 0 & 0 & 0 \\
0 & 0 & 0 & 0 & 0 & 0 & 0 & 0 & 0 & 0 & 128 & 0 & 0 & 0 \\
0 & 0 & 0 & 0 & 0 & 0 & 0 & 0 & 0 & 0 & 0 & 128 & 0 & 0 \\
0 & 0 & 0 & 0 & 0 & 0 & 0 & 0 & 0 & 0 & 0 & 0 & 128 & 0 \\
0 & 0 & 0 & 0 & 0 & 0 & 0 & 0 & 0 & 0 & 0 & 0 & 0 & 128 \\
\end{array}
    \end{array}
\end{bmatrix}}
\]
\textbullet\ Reduce the basis of the lattice $L_{{\textbf{h}} }$ by applying some  basis reduction algorithm such as LLL or BKZ. Since we are dealing with small dimension, LLL is a good choice to apply here, and the reduced matrix is

\[
\scalemath{0.425}{
  M_{{\textbf{h}} }^{LLL}=
  \begin{bmatrix}
  \begin{array}{rrrrrrrrrrrrrrrrrrrrrrrrrrrr}
      
 -1 & 1 & 0 & 1 & -1 & -1 & 0 & 1 & -1 & 0 & 1 & 0 & -1 & 0 & -1 & 1 & -1 & 0 & 1 & 1 & 0 & 0 & 0 & -1 & 0 & 0 & -1 & 1 \\
-1 & 0 & 1 & 0 & -1 & 0 & 1 & 0 & -1 & 1 & 0 & 1 & -1 & -1 & 0 & -1 & 0 & 0 & -1 & 1 & 0 & 0 & -1 & 1 & -1 & 0 & 1 & 1 \\
-1 & 1 & 0 & -1 & 0 & 1 & 0 & 1 & -1 & 0 & -1 & 1 & 1 & 0 & 0 & 0 & 1 & 0 & 0 & 1 & -1 & 1 & -1 & 1 & 0 & -1 & -1 & 0 \\
1 & 0 & 1 & -1 & -1 & 0 & -1 & 0 & 1 & -1 & 0 & 1 & 0 & -1 & 1 & -1 & 0 & 1 & 1 & 0 & -1 & 1 & 0 & 0 & -1 & 0 & 0 & -1 \\
0 & 1 & -1 & -1 & 0 & -1 & 1 & -1 & 0 & 1 & -1 & 0 & 1 & 0 & -1 & 0 & 1 & 1 & 0 & -1 & 1 & -1 & 1 & 0 & 0 & -1 & 0 & 0 \\
1 & -1 & -1 & 0 & -1 & 1 & 0 & 0 & -1 & 0 & 1 & -1 & 0 & 1 & 0 & 1 & 1 & 0 & -1 & 1 & -1 & 0 & -1 & 1 & 0 & 0 & -1 & 0 \\
1 & 1 & 0 & 1 & -1 & 0 & -1 & -1 & 0 & 1 & 0 & -1 & 1 & 0 & -1 & -1 & 0 & 1 & -1 & 1 & 0 & 0 & 0 & 1 & -1 & 0 & 0 & 1 \\
0 & 1 & -1 & 0 & -1 & 1 & 1 & 1 & 0 & -1 & 0 & 1 & 0 & -1 & 0 & 1 & -1 & 1 & 0 & -1 & -1 & 0 & 1 & 0 & 0 & 1 & -1 & 0 \\
-1 & 0 & -1 & 1 & 0 & 1 & -1 & 0 & 1 & 0 & -1 & 0 & 1 & -1 & 1 & 0 & -1 & 1 & -1 & 0 & 1 & -1 & 0 & 0 & -1 & 1 & 0 & 0 \\
1 & 1 & 1 & 1 & 1 & 1 & 1 & 1 & 1 & 1 & 1 & 1 & 1 & 1 & 0 & 0 & 0 & 0 & 0 & 0 & 0 & 0 & 0 & 0 & 0 & 0 & 0 & 0 \\
0 & -1 & 0 & 1 & 0 & -1 & 1 & 1 & 0 & 1 & -1 & 0 & -1 & 1 & 1 & 0 & 0 & 1 & -1 & 0 & 0 & -1 & 0 & 1 & -1 & 1 & 0 & -1 \\
1 & 1 & -1 & -1 & 1 & 0 & -1 & -2 & -1 & -1 & 0 & 1 & 1 & 0 & -1 & 0 & -1 & 0 & 1 & 0 & -1 & 2 & 1 & -1 & 0 & 0 & -1 & 1 \\
0 & -1 & 1 & 0 & -1 & 0 & 1 & -1 & 0 & -1 & 1 & 1 & 0 & 1 & -1 & 0 & 0 & 1 & 0 & 0 & 1 & -1 & 1 & 0 & -1 & -1 & 0 & 1 \\
1 & -1 & 0 & 1 & -1 & -1 & 1 & -1 & -1 & 0 & 2 & 1 & 1 & 0 & 0 & -1 & 0 & 1 & 1 & 0 & 1 & 0 & 1 & -1 & -2 & -1 & 1 & 0 \\
4 & 0 & 0 & -4 & -1 & -8 & 11 & 0 & 0 & -4 & -1 & -8 & 11 & 4 & 10 & 0 & -24 & 9 & 7 & -7 & 5 & 0 & -24 & 9 & 7 & -7 & 5 & 10 \\
-6 & 10 & 5 & -1 & 1 & -6 & -3 & -6 & -1 & -8 & 10 & 5 & 0 & 1 & -6 & 3 & 10 & 1 & -23 & 10 & 6 & 11 & 7 & -7 & 3 & 10 & 0 & -25 \\
-17 & -7 & -1 & -2 & -2 & -8 & -10 & 3 & 9 & 11 & 18 & 8 & 2 & 3 & -7 & -5 & 5 & 6 & -18 & -7 & -1 & 18 & 7 & 1 & 7 & 5 & -5 & -6 \\
-9 & -4 & 6 & 3 & 14 & -9 & 3 & -4 & 7 & 2 & 15 & -9 & 2 & -10 & -4 & 15 & -12 & 16 & -2 & -11 & -3 & 14 & -10 & 15 & -2 & -9 & -3 & -4 \\
1 & 4 & 0 & 8 & -10 & -5 & 0 & -12 & -5 & 0 & -1 & 5 & 1 & 9 & 24 & -9 & -8 & 8 & -5 & -10 & -1 & -4 & -9 & 0 & 23 & -8 & -8 & 7 \\
-14 & 9 & -4 & 10 & 6 & -4 & -2 & -14 & 9 & -4 & 10 & 6 & -4 & -2 & 3 & 11 & 2 & 4 & -14 & 10 & -16 & 3 & 11 & 2 & 4 & -14 & 10 & -16 \\
-1 & -2 & -2 & -8 & -10 & -17 & -7 & 2 & 3 & 3 & 9 & 11 & 18 & 8 & 5 & 6 & -18 & -7 & -1 & -7 & -5 & -5 & -6 & 18 & 7 & 1 & 7 & 5 \\
-4 & -11 & 1 & -5 & 3 & 13 & 2 & 1 & 0 & 20 & 2 & -12 & -8 & -5 & 12 & -5 & -20 & 8 & 6 & 2 & -4 & 6 & 15 & 4 & 4 & -24 & -10 & 6 \\
-15 & -2 & 12 & 3 & 3 & -9 & 8 & -13 & -4 & -2 & 10 & -11 & 16 & 2 & 5 & -4 & 1 & 18 & 2 & -11 & -13 & 1 & -21 & -2 & 13 & 12 & -5 & 4 \\
-3 & -3 & 4 & 3 & 4 & 6 & -10 & 2 & -5 & 7 & 1 & -4 & -6 & 6 & -1 & 11 & 1 & 18 & 9 & -33 & -5 & 8 & -18 & 4 & -8 & -9 & 9 & 14 \\
-3 & 11 & -19 & 1 & 3 & 16 & -11 & 4 & -2 & -3 & 2 & 5 & -22 & 15 & -10 & -14 & 16 & -1 & 11 & -2 & 1 & -10 & 5 & 0 & -11 & 6 & 10 & -1 \\
-15 & 11 & 2 & -11 & 20 & -2 & -3 & 2 & -3 & -5 & 21 & -14 & -4 & 3 & 2 & -1 & 9 & 15 & -16 & 1 & -12 & 1 & 12 & -6 & -11 & 2 & 9 & -5 \\
-14 & -4 & 1 & 4 & -2 & -6 & 22 & -11 & 18 & 0 & -2 & -16 & 12 & 2 & 2 & 9 & -5 & 0 & 12 & -6 & -10 & 15 & -16 & 0 & -12 & 2 & 0 & 9 \\
10 & -3 & -3 & -12 & -5 & -10 & -5 & 13 & -1 & 2 & -4 & 4 & 8 & 3 & -2 & -1 & -26 & -12 & -22 & -6 & -6 & -8 & -9 & -9 & -4 & -21 & 0 & -2 \\
    \end{array}

\end{bmatrix}}
\]

\textbullet\ Algorithm ~\ref{alg:naive_approach} returns $ \normalfont{\textbf{k}} = (\normalfont{\textbf{f}}^{\prime}, \normalfont{\textbf{g}}^{\prime})$ (the first row of the matrix  $M_{{\textbf{h}} }^{LLL}$) as a solution to the SVP, where
\begin{align*}
 \normalfont{\textbf{f}}^{\prime} &=  (-1, 1, 0, 1, -1, -1, 0, 1, -1, 0, 1, 0, -1, 0)\\ 
\normalfont{\textbf{g}}^{\prime} &= (-1, 1, -1, 0, 1, 1, 0, 0, 0, -1, 0, 0, -1, 1). 
\end{align*}
We can notice that the returned key $(\normalfont{\textbf{f}}^{\prime}, \normalfont{\textbf{g}}^{\prime})$ is different from the actual key. However, it has the same norm and since $\normalfont{\textbf{f}}^{\prime}$ is invertible in $ \Z_p D_N $, it can be used to decrypt the ciphertext and retrieve the message.\\
As we have noticed, the naive approach retrieved the private key by solving the SVP for a matrix of dimension $28 \times 28$. However, our contribution shows how to retrieve a decryption key by solving two instances of the SVP in matrices of dimensions $14 \times 14$.

\subsection{The pull-back approach}
The pull-back approach tries to retrieve two decryption keys; one of them short-enough and can serve as a decryption key, and the other is a ternary decryption key(returns with high probability for large $N$). 
The  steps of this approach are mentioned in Algorithm ~\ref{alg:ternary key}.

\textbullet\ Build two matrices $M_{{\textbf{h}}_0+{\textbf{h}}_1}$, $M_{{\textbf{h}}_0-{\textbf{h}}_1}$ for the  lattices $L_{{\textbf{h}}_0+{\textbf{h}}_1}$, $L_{{\textbf{h}}_0+{\textbf{h}}_1}$, respectively.
\newcolumntype{C}[1]{>{ \centering\arraybackslash$}m{#1}<{$}}
\newlength{\mycolwd}                                         
\newlength{\mycolwdd}
\settowidth{\mycolwd}{$-99.$}%
\settowidth{\mycolwdd}{$--$}%

 \[
\scalemath{0.7}{
M_{{\textbf{h}}_0+{\textbf{h}}_1}=
   \begin{bmatrix}
  \begin{array}{*{7}{@{}C{\mycolwdd}@{}} | *{7}{@{}C{\mycolwd}@{}}}
  
1 & 0 & 0 & 0 & 0 & 0 & 0 &144 & 150 & 151 & 180 & 78 & 39 & 154\\
0 & 1 & 0 & 0 & 0 & 0 & 0 & 161 & 149 & 114 & 122 & 144 & 174 & 32 \\
0 & 0 & 1 & 0 & 0 & 0 & 0 & 37 & 125 & 120 & 78 & 218 & 137 & 181\\
0 & 0 & 0 & 1 & 0 & 0 & 0 & 145 & 8 & 89 & 216 & 71 & 225 & 142\\
0 & 0 & 0 & 0 & 1 & 0 & 0 & 113 & 109 & 104 & 82 & 223 & 76 & 189  \\
0 & 0 & 0 & 0 & 0 & 1 & 0 & 153 & 209 & 102 & 111 & 87 & 187 & 47\\
0 & 0 & 0 & 0 & 0 & 0 & 1 &143 & 146 & 216 & 107 & 75 & 58 & 151  \\
 \hline

0 & 0 & 0 & 0 & 0 & 0 & 0 & 128 & 0 & 0 & 0 & 0 & 0 & 0\\
0 & 0 & 0 & 0 & 0 & 0 & 0 & 0 & 128 & 0 & 0 & 0 & 0 & 0\\
0 & 0 & 0 & 0 & 0 & 0 & 0 & 0 & 0 & 128 & 0 & 0 & 0 & 0\\
0 & 0 & 0 & 0 & 0 & 0 & 0 & 0 & 0 & 0 & 128 & 0 & 0 & 0\\
0 & 0 & 0 & 0 & 0 & 0 & 0 & 0 & 0 & 0 & 0 & 128 & 0 & 0\\
0 & 0 & 0 & 0 & 0 & 0 & 0 & 0 & 0 & 0 & 0 & 0 & 128 & 0\\
0 & 0 & 0 & 0 & 0 & 0 & 0 & 0 & 0 & 0 & 0 & 0 & 0 & 128\\

 \end{array}

\end{bmatrix}}
\]

\[
\scalemath{0.7}{
M_{{\textbf{h}}_0-{\textbf{h}}_1}=
   \begin{bmatrix}
  \begin{array}{*{7}{@{}C{\mycolwdd}@{}} | *{7}{@{}C{\mycolwd}@{}}} 
1 & 0 & 0 & 0 & 0 & 0 & 0 & 86 & -66 & 83 & 36 & 68 & -33 & -48 \\
0 & 1 & 0 & 0 & 0 & 0 & 0 & -55 & 81 & -30 & 112 & 72 & -28 & -26 \\
0 & 0 & 1 & 0 & 0 & 0 & 0 & -31 & -19 & 110 & 6 & 16 & 79 & -35\\
0 & 0 & 0 & 1 & 0 & 0 & 0 & 1 & -2 & 17 & 14 & 13 & 9 & 74 \\
0 & 0 & 0 & 0 & 1 & 0 & 0 & 103 & 37 & -98 & 24 & 7 & 8 & 45\\
0 & 0 & 0 & 0 & 0 & 1 & 0 & 81 & 7 & 44 & -105 & 19 & 43 & 37\\
0 & 0 & 0 & 0 & 0 & 0 & 1 & -59 & 88 & 0 & 39 & -69 & 48 & 79 \\
 \hline

0 & 0 & 0 & 0 & 0 & 0 & 0 & 128 & 0 & 0 & 0 & 0 & 0 & 0 \\
0 & 0 & 0 & 0 & 0 & 0 & 0 & 0 & 128 & 0 & 0 & 0 & 0 & 0\\
0 & 0 & 0 & 0 & 0 & 0 & 0 & 0 & 0 & 128 & 0 & 0 & 0 & 0\\
0 & 0 & 0 & 0 & 0 & 0 & 0 & 0 & 0 & 0 & 128 & 0 & 0 & 0\\
0 & 0 & 0 & 0 & 0 & 0 & 0 & 0 & 0 & 0 & 0 & 128 & 0 & 0\\
0 & 0 & 0 & 0 & 0 & 0 & 0 & 0 & 0 & 0 & 0 & 0 & 128 & 0\\
0 & 0 & 0 & 0 & 0 & 0 & 0 & 0 & 0 & 0 & 0 & 0 & 0 & 128\\

    \end{array}
\end{bmatrix}}
\]

\textbullet\ Apply LLL algorithm to the lattices  $L_{{\textbf{h}}_0+{\textbf{h}}_1}$, $L_{{\textbf{h}}_0-{\textbf{h}}_1}$ to get

\[
\scalemath{0.7}{
  M_{{\textbf{h}}_0+{\textbf{h}}_1}^{LLL}=
  \begin{bmatrix}
  \begin{array}{rrrrrrrrrrrrrr}
      
 1 & 1 & 1 & 1 & 1 & 1 & 1 & 0 & 0 & 0 & 0 & 0 & 0 & 0 \\
1 & -1 & 1 & 0 & 0 & -2 & 2 & 0 & 0 & 1 & 0 & 0 & 0 & -1 \\
0 & 1 & 1 & 1 & -2 & 1 & -1 & -1 & -1 & 1 & 0 & -1 & 1 & 1 \\
-2 & -1 & 0 & -1 & 2 & 0 & -1 & 1 & -1 & -1 & -1 & 0 & 1 & 1 \\
1 & -2 & -1 & 1 & -2 & 1 & 1 & 0 & 0 & 2 & 0 & -1 & 0 & -1 \\
0 & -2 & 0 & 3 & 0 & 0 & 1 & 0 & 0 & -1 & -1 & 0 & 1 & 1 \\
-1 & 1 & 0 & -2 & 0 & 0 & 1 & -2 & 1 & 1 & 1 & -1 & -1 & 1 \\
-2 & -7 & 10 & 6 & -1 & 0 & -6 & 6 & -6 & 2 & 10 & 1 & -24 & 11 \\
-14 & 9 & -4 & 10 & 6 & -4 & -2 & 3 & 11 & 2 & 4 & -14 & 10 & -16 \\
8 & -8 & 8 & -5 & 6 & 4 & -15 & -12 & 12 & 11 & -21 & 14 & -7 & 3 \\
-1 & -1 & 10 & -8 & -11 & 5 & 3 & 8 & 19 & -19 & 4 & -19 & 14 & -7 \\
-18 & -5 & 9 & 3 & -1 & 10 & 3 & -15 & -5 & 7 & 1 & 27 & 0 & -15 \\
3 & 11 & 1 & -7 & -15 & 0 & 7 & 19 & 9 & 4 & -33 & -1 & -13 & 15 \\
6 & -7 & 7 & -10 & -5 & 8 & 0 & -27 & -13 & -30 & -14 & -15 & -7 & -22 \\
    \end{array}

\end{bmatrix}}
\]

\[
\scalemath{0.73}{
  M_{{\textbf{h}}_0-{\textbf{h}}_1}^{LLL}=
  \begin{bmatrix}
  \begin{array}{rrrrrrrrrrrrrr}
1 & 0 & -1 & -1 & 0 & 1 & -1 & 0 & 2 & 0 & 0 & -1 & 2 & -1 \\
-1 & -1 & 2 & 0 & 0 & 2 & -1 & 0 & 1 & -1 & -1 & -1 & 1 & -1 \\
-1 & 1 & 0 & 0 & -2 & 1 & 2 & 0 & 0 & -1 & 1 & -1 & 0 & -1 \\
2 & -2 & 0 & 0 & 1 & 0 & 0 & 1 & -1 & 0 & 0 & -1 & -2 & 1 \\
-2 & -1 & 1 & -1 & 0 & 1 & 1 & 1 & 1 & 1 & -2 & 1 & -1 & 1 \\
1 & -1 & 2 & -1 & -2 & 0 & 0 & 0 & -1 & 0 & 2 & 1 & 0 & 0 \\
1 & 1 & 1 & -2 & 0 & 0 & 0 & -2 & 0 & 1 & -2 & 2 & 0 & -1 \\
-4 & -2 & -4 & -7 & -9 & -19 & -7 & 6 & 6 & -17 & -8 & 0 & -8 & -3 \\
7 & 15 & 11 & 7 & 8 & 8 & -4 & -1 & 18 & 4 & 12 & 3 & -15 & 3 \\
2 & -14 & -3 & -3 & 10 & -9 & 16 & 5 & 0 & -19 & -3 & 12 & 11 & -4 \\
8 & -1 & 2 & 11 & -18 & -4 & 1 & -9 & 5 & 1 & -15 & -8 & 9 & 19 \\
-2 & 0 & 15 & -6 & 12 & -25 & 6 & 6 & 16 & 2 & 9 & -11 & -6 & -16 \\
0 & 10 & 10 & -16 & 7 & -17 & 4 & 24 & -8 & -10 & -7 & -6 & 11 & 0 \\
3 & -11 & -2 & 19 & -9 & -6 & 5 & 4 & 0 & 20 & -5 & 1 & 4 & -22 \\
    \end{array}

\end{bmatrix}}
\]
\textbullet\ Algorithm ~\ref{alg:ternary key} finds $\normalfont{\textbf{k}_1} = (\normalfont{\textbf{f}}^{\prime}, \normalfont{\textbf{g}}^{\prime})$ where,

\begin{align*}
\normalfont{\textbf{f}}^{\prime} &=  (2, -1, 0, -1, 0, -1, 1, 0, -1, 2, 1, 0, -3, 3)\\ 
\normalfont{\textbf{g}}^{\prime} &= (0, 2, 1, 0, -1, 2, -2, 0, -2, 1, 0, 1, -2, 0).  
\end{align*}
We can notice that the returned key $(\normalfont{\textbf{f}}^{\prime}, \normalfont{\textbf{g}}^{\prime})$ is not a ternary key. However, it has small norm $\norm{(\normalfont{\textbf{f}}^{\prime}, \normalfont{\textbf{g}}^{\prime})} = \sqrt{56} < 2 \times  \norm{(\normalfont{\textbf{f}}, \normalfont{\textbf{g}})}$, so one can compute ${\textbf{a}}^\prime \equiv {\textbf{f}}^\prime\star {\textbf{c}}^\prime \pmod q$ and after centerlifting it gets,  ${\textbf{a}}^\prime_{lifted} = (-18, 4, 0, -19, 12, 0, 5, 8, -6,\\ -12, 23, -16, 0, 11)$ (which equals and not merely congruent modulo $q$ to $p{\textbf{g}}^\prime\star{\textbf{r}} + {\textbf{f}}^\prime\star {\textbf{m}}$ )
and since $\normalfont{\textbf{f}}^{\prime}$ is invertible in $ \Z_p D_N $, the attacker calculates ${\textbf{f}}^\prime_p(x)\star {\textbf{a}}^\prime_{lifted} \pmod p$ and centerlifting  it again modulo $p$, he retrieves the  message correctly.\\

\textbullet\ Pull-back approach tries to retrieve better decryption key (ternary key), so Algorithm ~\ref{alg:ternary key} does further processing and search for a match of two rotated vectors that enables retrieving a ternary key lying in $L_{{\textbf{h}}}$.
Algorithm ~\ref{alg:ternary key} finds 
$({\textbf{f}}_0^{(-1)}+{\textbf{f}}_1^{(1)},{\textbf{g}}_0^{(-1)}+{\textbf{g}}_1^{(1)}) = \left( (1, -1, 1, 0, 0, -2, 2),( 0,   0,  1,   0,  0,  0,  -1)\right)$ in lattice $ L_{{\textbf{h}}_0+{\textbf{h}}_1}$, $({\textbf{f}}_0^{(-1)}-{\textbf{f}}_1^{(1)},{\textbf{g}}_0^{(-1)}-{\textbf{g}}_1^{(1)})$ 
$= \left( (1, 1, 1, -2, 0, 0, 0), (-2,   0,   1,  -2,   2,   0 , -1) \right)$
in lattice $L_{{\textbf{h}}_0-{\textbf{h}}_1}$.

\noindent
Finally, it returns the ternary key $ \normalfont{\textbf{k}}_2 = (\normalfont{\textbf{f}}^{\prime\prime},{\textbf{g}}^{\prime\prime}) = \left(({\textbf{f}}_0^{(-1)},{\textbf{f}}_1^{(1)}),({\textbf{g}}_0^{(-1)},{\textbf{g}}_1^{(1)})\right)\in$  $L_{{\textbf{h}}}$ with,
\noindent
\begin{align*}
    \normalfont{\textbf{f}}^{\prime\prime} & = (1, 0, 1, -1, 0, -1, 1, 0, -1, 0, 1, 0, -1, 1)\\
     \normalfont{\textbf{g}^{\prime\prime}} &= (-1, 0, 1, -1, 1, 0, -1, 1, 0, 0, 1, -1, 0, 0).
\end{align*}
The returned key is not exactly the private key, but it is ternary with the same norm and since $\normalfont{\textbf{f}}^{\prime\prime}$ is invertible in $\Z_p D_N$, it can be used perfectly to decrypt the message.
\end{appendices}
\bibliography{sn-bibliography}


\begin{thebibliography}{28}
\ifx \bisbn   \undefined \def \bisbn  #1{ISBN #1}\fi
\ifx \binits  \undefined \def \binits#1{#1}\fi
\ifx \bauthor  \undefined \def \bauthor#1{#1}\fi
\ifx \batitle  \undefined \def \batitle#1{#1}\fi
\ifx \bjtitle  \undefined \def \bjtitle#1{#1}\fi
\ifx \bvolume  \undefined \def \bvolume#1{\textbf{#1}}\fi
\ifx \byear  \undefined \def \byear#1{#1}\fi
\ifx \bissue  \undefined \def \bissue#1{#1}\fi
\ifx \bfpage  \undefined \def \bfpage#1{#1}\fi
\ifx \blpage  \undefined \def \blpage #1{#1}\fi
\ifx \burl  \undefined \def \burl#1{\textsf{#1}}\fi
\ifx \doiurl  \undefined \def \doiurl#1{\url{https://doi.org/#1}}\fi
\ifx \betal  \undefined \def \betal{\textit{et al.}}\fi
\ifx \binstitute  \undefined \def \binstitute#1{#1}\fi
\ifx \binstitutionaled  \undefined \def \binstitutionaled#1{#1}\fi
\ifx \bctitle  \undefined \def \bctitle#1{#1}\fi
\ifx \beditor  \undefined \def \beditor#1{#1}\fi
\ifx \bpublisher  \undefined \def \bpublisher#1{#1}\fi
\ifx \bbtitle  \undefined \def \bbtitle#1{#1}\fi
\ifx \bedition  \undefined \def \bedition#1{#1}\fi
\ifx \bseriesno  \undefined \def \bseriesno#1{#1}\fi
\ifx \blocation  \undefined \def \blocation#1{#1}\fi
\ifx \bsertitle  \undefined \def \bsertitle#1{#1}\fi
\ifx \bsnm \undefined \def \bsnm#1{#1}\fi
\ifx \bsuffix \undefined \def \bsuffix#1{#1}\fi
\ifx \bparticle \undefined \def \bparticle#1{#1}\fi
\ifx \barticle \undefined \def \barticle#1{#1}\fi
\bibcommenthead
\ifx \bconfdate \undefined \def \bconfdate #1{#1}\fi
\ifx \botherref \undefined \def \botherref #1{#1}\fi
\ifx \url \undefined \def \url#1{\textsf{#1}}\fi
\ifx \bchapter \undefined \def \bchapter#1{#1}\fi
\ifx \bbook \undefined \def \bbook#1{#1}\fi
\ifx \bcomment \undefined \def \bcomment#1{#1}\fi
\ifx \oauthor \undefined \def \oauthor#1{#1}\fi
\ifx \citeauthoryear \undefined \def \citeauthoryear#1{#1}\fi
\ifx \endbibitem  \undefined \def \endbibitem {}\fi
\ifx \bconflocation  \undefined \def \bconflocation#1{#1}\fi
\ifx \arxivurl  \undefined \def \arxivurl#1{\textsf{#1}}\fi
\csname PreBibitemsHook\endcsname

\bibitem{ntru-first-paper}
\begin{bchapter}
\bauthor{\bsnm{Hoffstein}, \binits{J.}},
\bauthor{\bsnm{Pipher}, \binits{J.}},
\bauthor{\bsnm{Silverman}, \binits{J.H.}}:
\bctitle{Ntru: A ring-based public key cryptosystem}.
In: \bbtitle{International Algorithmic Number Theory Symposium},
\bconflocation{Berlin, Heidelberg},
pp. \bfpage{267}--\blpage{288}
(\byear{1998}).
\doiurl{10.1007/BFb0054868}.
\bcomment{Springer}
\end{bchapter}
\endbibitem

\bibitem{IEEE}
\begin{botherref}
\oauthor{\bsnm{{Working Group of the C/MM Committee and others}}}:
IEEE P1363.1 Standard Specification for Public-Key Cryptographic Techniques
  Based on Hard Problems over Lattices
(2009)
\end{botherref}
\endbibitem

\bibitem{NTRUEncrypt}
\begin{botherref}
\oauthor{\bsnm{Chen}, \binits{C.}},
\oauthor{\bsnm{Hoffstein}, \binits{J.}},
\oauthor{\bsnm{Whyte}, \binits{W.}},
\oauthor{\bsnm{Zhang}, \binits{Z.}}:
{NIST PQ} submission: {NTRUE}ncrypt a lattice based encryption algorithm.
NIST
(2017).
\url{https://csrc.nist.gov/Projects/post-quantum-cryptography/post-quantum-cryptography-standardization/round-1-submissions}
\end{botherref}
\endbibitem

\bibitem{pqNTRU}
\begin{botherref}
\oauthor{\bsnm{Chen}, \binits{C.}},
\oauthor{\bsnm{Hoffstein}, \binits{J.}},
\oauthor{\bsnm{Whyte}, \binits{W.}},
\oauthor{\bsnm{Zhang}, \binits{Z.}}:
{NIST PQ Submission}: pq{NTRUS}ign a modular lattice signature scheme.
NIST
(2017).
\url{https://csrc.nist.gov/Projects/post-quantum-cryptography/post-quantum-cryptography-standardization/round-1-submissions}
\end{botherref}
\endbibitem

\bibitem{HRSS-round1}
\begin{botherref}
\oauthor{\bsnm{Hülsing}, \binits{A.}},
\oauthor{\bsnm{Rijneveld}, \binits{J.}},
\oauthor{\bsnm{Schanck}, \binits{J.M.}},
\oauthor{\bsnm{Schwabe}, \binits{P.}}:
{NIST PQ Submission}: {NTRU-HRSS-KEM}.
NIST
(2017).
\url{https://csrc.nist.gov/Projects/post-quantum-cryptography/post-quantum-cryptography-standardization/round-1-submissions}
\end{botherref}
\endbibitem

\bibitem{NTRUprime-round1}
\begin{botherref}
\oauthor{\bsnm{Bernstein}, \binits{D.J.}},
\oauthor{\bsnm{Chuengsatiansup}, \binits{C.}},
\oauthor{\bsnm{Lange}, \binits{T.}},
\oauthor{\bsnm{Vredendaal}, \binits{C.v.}}:
{NIST PQ Submission}: {NTRU} prime.
NIST
(2017).
\url{https://csrc.nist.gov/Projects/post-quantum-cryptography/post-quantum-cryptography-standardization/round-1-submissions}
\end{botherref}
\endbibitem

\bibitem{NTRUprime-round2}
\begin{botherref}
\oauthor{\bsnm{Bernstein}, \binits{D.J.}},
\oauthor{\bsnm{Chuengsatiansup}, \binits{C.}},
\oauthor{\bsnm{Lange}, \binits{T.}},
\oauthor{\bsnm{Vredendaal}, \binits{C.v.}}:
{NIST PQ Submission}: {NTRU} prime.
NIST
(2019).
\url{https://csrc.nist.gov/Projects/post-quantum-cryptography/post-quantum-cryptography-standardization/round-2-submissions}
\end{botherref}
\endbibitem

\bibitem{NTRU-round2}
\begin{botherref}
\oauthor{\bsnm{Chen}, \binits{C.}},
\oauthor{\bsnm{Danba}, \binits{O.}},
\oauthor{\bsnm{Hoffsstein}, \binits{J.}},
\oauthor{\bsnm{Hülsing}, \binits{A.}},
\oauthor{\bsnm{Rijneveld}, \binits{J.}},
\oauthor{\bsnm{M.~Schanck}, \binits{J.}},
\oauthor{\bsnm{Schwabe}, \binits{P.}},
\oauthor{\bsnm{Whyte}, \binits{W.}},
\oauthor{\bsnm{Zhang}, \binits{Z.}}:
{NIST PQ Submission}: {NTRU Algorithm Specifications And Supporting
  Documentation}.
NIST
(2019).
\url{https://csrc.nist.gov/Projects/post-quantum-cryptography/post-quantum-cryptography-standardization/round-2-submissions}
\end{botherref}
\endbibitem

\bibitem{NTRU-round3}
\begin{botherref}
\oauthor{\bsnm{Chen}, \binits{C.}},
\oauthor{\bsnm{Danba}, \binits{O.}},
\oauthor{\bsnm{Hoffstein}, \binits{J.}},
\oauthor{\bsnm{H{\"u}lsing}, \binits{A.}},
\oauthor{\bsnm{Rijneveld}, \binits{J.}},
\oauthor{\bsnm{Schanck}, \binits{J.M.}},
\oauthor{\bsnm{Schwabe}, \binits{P.}},
\oauthor{\bsnm{Whyte}, \binits{W.}},
\oauthor{\bsnm{Zhang}, \binits{Z.}}:
Ntru:algorithm specifications and supporting documentation.
NIST
(2020).
\url{https://csrc.nist.gov/Projects/post-quantum-cryptography/post-quantum-cryptography-standardization/round-3-submissions}
\end{botherref}
\endbibitem

\bibitem{nist-third-report}
\begin{botherref}
\oauthor{\bsnm{Alagic}, \binits{G.}},
\oauthor{\bsnm{Apon}, \binits{D.}},
\oauthor{\bsnm{Cooper}, \binits{D.}},
\oauthor{\bsnm{Dang}, \binits{Q.}},
\oauthor{\bsnm{Dang}, \binits{T.}},
\oauthor{\bsnm{Kelsey}, \binits{J.}},
\oauthor{\bsnm{Lichtinger}, \binits{J.}},
\oauthor{\bsnm{Miller}, \binits{C.}},
\oauthor{\bsnm{Moody}, \binits{D.}},
\oauthor{\bsnm{Peralta}, \binits{R.}}, et al.:
Status report on the third round of the nist post-quantum cryptography
  standardization process.
US Department of Commerce, NIST
(2022).
\url{https://tsapps.nist.gov/publication/get_pdf.cfm?pub_id=934458}
\end{botherref}
\endbibitem

\bibitem{Coppersmith-Shamir-paper}
\begin{bchapter}
\bauthor{\bsnm{Coppersmith}, \binits{D.}},
\bauthor{\bsnm{Shamir}, \binits{A.}}:
\bctitle{{L}attice {A}ttacks on {NTRU}}.
In: \bbtitle{Advances in Cryptology --- EUROCRYPT '97},
pp. \bfpage{52}--\blpage{61}.
\bpublisher{Springer},
\blocation{Berlin, Heidelberg}
(\byear{1997}).
\doiurl{10.1007/3-540-69053-0_5}
\end{bchapter}
\endbibitem

\bibitem{non-commutative-ntru-unpublished}
\begin{botherref}
\oauthor{\bsnm{Hoffstein}, \binits{J.}},
\oauthor{\bsnm{Silverman}, \binits{J.H.}}:
A non-commutative version of the ntru public key cryptosystem.
unpublished paper, February
(1997)
\end{botherref}
\endbibitem

\bibitem{ibm-report}
\begin{botherref}
\oauthor{\bsnm{Coppersmith}, \binits{D.}}:
Attacking non-commutative {NTRU}.
Technical report,
Technical report, {IBM} research report, April 1997. Report
(2006).
\url{https://dominoweb.draco.res.ibm.com/d102d0885e971b558525659300727a26.html}
\end{botherref}
\endbibitem

\bibitem{Marylandthesis}
\begin{botherref}
\oauthor{\bsnm{Truman}, \binits{K.R.}}:
Analysis and extension of non-commutative {NTRU}.
{PhD} dissertation,
University of Maryland
(2007).
\url{https://drum.lib.umd.edu/handle/1903/7344}
\end{botherref}
\endbibitem

\bibitem{QTRU}
\begin{botherref}
\oauthor{\bsnm{Malekian}, \binits{E.}},
\oauthor{\bsnm{Zakerolhosseini}, \binits{A.}},
\oauthor{\bsnm{Mashatan}, \binits{A.}}:
Qtru : a lattice attack resistant version of ntru pkcs based on quaternion
  algebra.
IACR Cryptology ePrint Archive
\textbf{2009}
(2009)
\end{botherref}
\endbibitem

\bibitem{Sakurai}
\begin{botherref}
\oauthor{\bsnm{Yasuda}, \binits{T.}},
\oauthor{\bsnm{Dahan}, \binits{X.}},
\oauthor{\bsnm{Sakurai}, \binits{K.}}:
Characterizing ntru-variants using group ring and evaluating their lattice
  security.
{IACR} Cryptol. ePrint Arch.,
1170
(2015)
\end{botherref}
\endbibitem

\bibitem{Gentry}
\begin{bchapter}
\bauthor{\bsnm{Gentry}, \binits{C.}}:
\bctitle{Key recovery and message attacks on ntru-composite}.
In: \beditor{\bsnm{Pfitzmann}, \binits{B.}} (ed.)
\bbtitle{Advances in Cryptology --- EUROCRYPT 2001},
pp. \bfpage{182}--\blpage{194}.
\bpublisher{Springer},
\blocation{Berlin, Heidelberg}
(\byear{2001}).
\burl{https://doi.org/10.1007/3-540-44987-6_12}
\end{bchapter}
\endbibitem

\bibitem{MS}
\begin{bbook}
\bauthor{\bsnm{Milies}, \binits{C.}},
\bauthor{\bsnm{Sehgal}, \binits{S.}}:
\bbtitle{An Introduction to Group Rings},
(\byear{2002}).
\doiurl{10.1007/978-94-010-0405-3}
\end{bbook}
\endbibitem

\bibitem{Hurley}
\begin{barticle}
\bauthor{\bsnm{Hurley}, \binits{T.}}:
\batitle{Group rings and rings of matrices}.
\bjtitle{International Journal of Pure and Applied Mathematics}
\bvolume{31},
\bfpage{319}--\blpage{335}
(\byear{2006})
\end{barticle}
\endbibitem

\bibitem{HPS}
\begin{bbook}
\bauthor{\bsnm{Hoffstein}, \binits{J.}},
\bauthor{\bsnm{Pipher}, \binits{J.}},
\bauthor{\bsnm{Silverman}, \binits{J.H.}}:
\bbtitle{An Introduction to Mathematical Cryptography},
\bedition{1}st edn.
\bpublisher{Springer},
\blocation{NY}
(\byear{2008})
\end{bbook}
\endbibitem

\bibitem{LLL}
\begin{barticle}
\bauthor{\bsnm{Lenstra}, \binits{A.K.}},
\bauthor{\bsnm{Lenstra}, \binits{H.W.}},
\bauthor{\bsnm{Lov{\'a}sz}, \binits{L.}}:
\batitle{Factoring polynomials with rational coefficients}.
\bjtitle{Mathematische annalen}
\bvolume{261}(\bissue{ARTICLE}),
\bfpage{515}--\blpage{534}
(\byear{1982}).
\doiurl{10.1007/BF01457454}
\end{barticle}
\endbibitem

\bibitem{BKZ}
\begin{barticle}
\bauthor{\bsnm{Schnorr}, \binits{C.P.}}:
\batitle{A hierarchy of polynomial time lattice basis reduction algorithms}.
\bjtitle{Theoretical computer science}
\bvolume{53}(\bissue{2-3}),
\bfpage{201}--\blpage{224}
(\byear{1987}).
\doiurl{10.1016/0304-3975(87)90064-8}
\end{barticle}
\endbibitem

\bibitem{BKZ2.0}
\begin{bchapter}
\bauthor{\bsnm{Chen}, \binits{Y.}},
\bauthor{\bsnm{Nguyen}, \binits{P.Q.}}:
\bctitle{Bkz 2.0: Better lattice security estimates}.
In: \bbtitle{International Conference on the Theory and Application of
  Cryptology and Information Security},
pp. \bfpage{1}--\blpage{20}
(\byear{2011}).
\doiurl{10.1007/978-3-642-25385-0_1}.
\bcomment{Springer}
\end{bchapter}
\endbibitem

\bibitem{progressive_bkz}
\begin{bchapter}
\bauthor{\bsnm{Aono}, \binits{Y.}},
\bauthor{\bsnm{Wang}, \binits{Y.}},
\bauthor{\bsnm{Hayashi}, \binits{T.}},
\bauthor{\bsnm{Takagi}, \binits{T.}}:
\bctitle{Improved progressive bkz algorithms and their precise cost estimation
  by sharp simulator}.
In: \bbtitle{Annual International Conference on the Theory and Applications of
  Cryptographic Techniques},
pp. \bfpage{789}--\blpage{819}
(\byear{2016}).
\doiurl{10.1007/978-3-662-49890-3_30}.
\bcomment{Springer}
\end{bchapter}
\endbibitem

\bibitem{deep_bkz}
\begin{bchapter}
\bauthor{\bsnm{Yamaguchi}, \binits{J.}},
\bauthor{\bsnm{Yasuda}, \binits{M.}}:
\bctitle{Explicit formula for gram-schmidt vectors in lll with deep insertions
  and its applications}.
In: \bbtitle{Number-Theoretic Methods in Cryptology: First International
  Conference, NuTMiC 2017, Warsaw, Poland, September 11-13, 2017, Revised
  Selected Papers},
pp. \bfpage{142}--\blpage{160}
(\byear{2018}).
\doiurl{10.1007/978-3-319-76620-1_9}.
\bcomment{Springer}
\end{bchapter}
\endbibitem

\bibitem{fplll}
\begin{botherref}
\oauthor{\bparticle{development} \bsnm{team}, \binits{T.F.}}:
{fplll}, a lattice reduction library, {Version}: 5.4.4.
Available at \url{https://github.com/fplll/fplll}
(2023)
\end{botherref}
\endbibitem

\bibitem{fpylll}
\begin{botherref}
\oauthor{\bparticle{development} \bsnm{team}, \binits{T.F.}}:
{fpylll}, a {Python} wraper for the {fplll} lattice reduction library,
  {Version}: 0.5.9.
Available at \url{https://github.com/fplll/fpylll}
(2023)
\end{botherref}
\endbibitem

\bibitem{mpfr}
\begin{barticle}
\bauthor{\bsnm{Fousse}, \binits{L.}},
\bauthor{\bsnm{Hanrot}, \binits{G.}},
\bauthor{\bsnm{Lef{\`e}vre}, \binits{V.}},
\bauthor{\bsnm{P{\'e}lissier}, \binits{P.}},
\bauthor{\bsnm{Zimmermann}, \binits{P.}}:
\batitle{Mpfr: A multiple-precision binary floating-point library with correct
  rounding}.
\bjtitle{ACM Transactions on Mathematical Software (TOMS)}
\bvolume{33}(\bissue{2}),
\bfpage{13}
(\byear{2007})
\end{barticle}
\endbibitem

\end{thebibliography}

%


\end{document}